\documentclass{article}
\usepackage{graphicx,xcolor,amsmath,amssymb,amsthm,hyperref,soul,cleveref,stmaryrd} 
\usepackage{comment}
\usepackage{circuitikz}
\usepackage{thmtools}
\usepackage{thm-restate}
\usepackage[margin=3.2cm]{geometry}
\usepackage{authblk}

\title{\textbf{A Linear Temporal Logic of Frequencies on Finite Series of Events}}

\author[1]{Melissa Antonelli}
\author[2]{Alessandro Buda}
\author[3]{Leonardo Ceragioli}
\author[4]{Ferdinando Fugazza Calda}
\author[5]{Giuseppe Primiero}

\affil[1]{Universit\"at T\"ubingen, Doblerstraße 33, T\"ubingen, 72074, Germany}
\affil[3,5]{University of Milan, Via Festa del Perdono 7, Milano, 20122, Italy}
\affil[2,4]{University School for Advanced Studies IUSS Pavia, Piazza della Vittoria n. 15, Pavia, 27100, Italy\vspace{1em}}
 
\affil[1]{Correspondence: melissa.antonelli@uni-tuebingen.de}
\affil[2]{Correspondence: alessandro.buda@iusspavia.it}
\affil[3]{Correspondence: leonardo.ceragioli@unimi.it}
\affil[4]{Correspondence: ferdinando.fugazza@iusspavia.it}
\affil[5]{Correspondence: giuseppe.primiero@unimi.it}

\date{}

\newtheorem{prop}{Proposition}[section]
\newtheorem{theorem}{Theorem}[section]
\newtheorem{lemma}{Lemma}[section]

\theoremstyle{definition}
\newtheorem{notation}{Notation}[section]
\newtheorem{definition}{Definition}[section]
\newtheorem{ex}{Example}[section]
\newtheorem{remark}{Remark}[section]

\newcommand{\tu}[1]{\mathbf{#1}}

\newcommand{\PL}{\mathsf{LTLFS}}
\newcommand{\Pg}{F_\Omega}
\newcommand{\Rexp}{R_e}

\newcommand{\str}{\mathbf{M}}
\newcommand{\model}{\mathbb{M}}
\newcommand{\vseg}{\mathbf{a}}

\newcommand{\Nat}{\mathbb{N}}
\newcommand{\at}{p}
\newcommand{\fOne}{\phi}
\newcommand{\fTwo}{\psi}
\newcommand{\fThree}{\gamma}

\newcommand{\At}{\mathfrak{P}}

\newcommand{\wbox}{\square}
\newcommand{\bstar}{\star}
\newcommand{\bbox}{\blacksquare}
\newcommand{\wnext}{\triangleright}

\newcommand{\sat}{\circledcirc}

\newcommand{\head}{\textsc{H}}
\newcommand{\tail}{\textsc{T}}

\newcommand{\fr}{\mu_{\Phi}}
\newcommand{\pr}{\mu}

\newcommand{\LTLf}{\textsf{LTL$_f$}}
\newcommand{\Until}{\mathcal{U}}

\begin{document}

\maketitle

\section{Introduction}

The need to monitor and control complex systems at runtime has become increasingly central in contemporary computational settings, especially when
outputs must satisfy quantitative requirements. 
This problem is closely connected to explainable artificial intelligence (XAI): beyond interpreting how individual outputs are produced, XAI strategies aim to provide transparent and auditable methods for assessing the behaviour of a system. Typically, this can be formulated across a temporally evolving sequence of decisions. A representative case is the evaluation of machine-learning classifiers with respect to algorithmic fairness. 
In such settings, a symbolic treatment allows to represent benchmark requirements explicitly, to verify whether the distribution of the outputs observed so far conforms to them, and to explain whether compliance remains achievable. This supports two complementary tasks: on the one hand, checking the current output distribution against a desirable or benchmark distribution; on the other, predicting whether the completed sequence will satisfy that benchmark. Such logically interpretable assessment enables corrective intervention before further computation is performed, and eventually an unacceptable final distribution is produced.

In this paper, we introduce $\PL$ (Linear Temporal Logic of Frequencies on Finite Series), a logical framework conservatively extending \LTLf\ (linear temporal logic over finite traces) designed to represent and reason about finite, temporally ordered series of probabilistic events whose occurrences determine an evolving frequency distribution on outputs. Our language contains measure-sensitive temporal operators that count the occurrences of an event at different stages of a finite series and compare the resulting empirical frequency with a specified target value. Other operators use the observations available at the current stage to evaluate whether the target frequency can, or must, still be attained by the end of the series. This logic therefore integrates within a single semantic structure two complementary forms of reasoning: the \emph{verification} of a given distribution against a desirable or benchmark distribution, and the \emph{prediction} of whether the benchmark will remain attainable or will eventually be satisfied. This combination constitutes the principal novelty of our framework. Our formalism allows the benchmark to be represented internally as part of the evaluation structure rather than being imposed externally after the logical analysis has been completed. 
Although the underlying mathematical objects are deliberately simple, the resulting distributional properties and their temporal evolution are not trivially expressible or verifiable in existing logical systems.

The features and expressivity of $\PL$ render it applicable to phenomena that are increasingly common in computational practice. 
For example, a machine-learning classifier may generate a finite sequence of decisions whose distribution across relevant categories is expected to satisfy a fairness requirement. Our logic can evaluate, at any stage of the sequence, whether the decisions produced so far conform to that requirement and whether the remaining decisions can still yield an acceptable final distribution. It can thus indicate when correction is unnecessary, when successful correction remains possible, and when violation of the benchmark has become unavoidable. Consider a generative AI system repeatedly prompted with the request
\textit{``Give me a picture of an engineer''}. For this toy example, assume that each image depicts one person and is assigned to exactly one of two categories, male or female, according to a fixed labeling scheme. Suppose that the reference training dataset has a distribution of 66\% male and 34\% female engineers in the images, while
successive output generations are independent and each produces a male engineer with probability 75\% and a female engineer with probability 25\%. This sampling assumption is introduced explicitly, since the proportion in the training dataset alone does not determine the distribution of generated outputs. Against these assumptions, the intended output consists instead of a predetermined number of images with equal representation of the two categories, which represents the benchmark against which the generated collection is evaluated. This requirement concerns the proportions in the completed collection and must be distinguished from the probabilities with which the system generates individual outputs. For example, in a series of $6$ outputs, the expected distribution should be $3$ males and $3$ females. The images are generated sequentially, and we follow their production
as a single ordered series. At each stage, we know the outputs already
produced and the number still to be generated. The central objective
is to use this partial information to assess the current distribution
and anticipate whether the completed series will satisfy the benchmark.

This scenario gives rise to five related problems:
\begin{enumerate}
\item \emph{Checking the compatibility of the current distribution against the target.}
Do the outputs generated so far match the intended proportions?
Suppose that the first four images depict, in order, a male engineer,
a female engineer, a female engineer, and a male engineer. The collection
is currently balanced. If the fifth image depicts a male engineer,
the observed proportions become 60\% male and 40\% female. The current
collection no longer satisfies equal representation, although the final
collection can still become balanced.

\item \emph{Identifying partial distributions permitted by the final
benchmark.}
Which distributions can occur at an intermediate stage of a series
that ultimately satisfies the requirement? In a balanced series of six
images, the first four can contain one, two, or three male engineers.
Their proportion can therefore be 25\%, 50\%, or 75\%. Four male
engineers would already exceed the number permitted in the final
collection, while four female engineers would create the corresponding
problem for the other category. This question identifies the intermediate
distributions allowed by the benchmark, independently of the particular
outputs observed.

\item \emph{Checking whether the observed series can still meet the
benchmark.}
Given the images already generated, can the remaining outputs restore
or preserve the required final distribution? In the given series of six generated output, if the first three images
all depict male engineers, balance remains achievable, but only if the
next three all depict female engineers. If the first four images depict
male engineers, balance is already impossible within the prescribed
six outputs. This failure can be detected before generation is complete.
The example thus distinguishes a temporary departure from the benchmark
from a departure that no remaining output can resolve.

\item \emph{Estimating the likelihood of a successful completion.}
When the benchmark is still attainable, how likely is it to be reached?
After the balanced four-image prefix described above, success requires
one male and one female engineer in the remaining two positions, in
either order. Under the assumed generation probabilities, the probability
of this happening is 37.5\%. If the two categories were instead generated
with equal probability, it would be 50\%. Similarly, after three male
outputs, balance remains possible but requires three consecutive female
outputs, whose probability under our assumptions is only one in
sixty-four. Attainability and likelihood therefore provide different
information for monitoring the generation process.

\item \emph{Predicting subsequent outputs conditional on meeting the
benchmark.}
Suppose that the completed collection is known, or hypothesized, to
satisfy the requirement. What does this imply about the next output?
After the balanced four-image prefix, the two successful continuations
are equally likely under our sampling assumptions, so either category
has conditional probability 50\% of appearing next. If the fifth image
then depicts a male engineer, the sixth must depict a female engineer
for the collection to remain compatible with the benchmark. This
conditional requirement does not change the generator's own behaviour:
without conditioning on successful completion, its probability of
producing a female engineer remains 25\%.
\end{enumerate}

The example motivates the need for both the verification and prediction tasks addressed by our logic.
Verification concerns the distribution observed
so far and its compatibility with the final requirement. Prediction concerns the likelihood of satisfying that requirement and the subsequent outcomes consistent with its realization. 
Together, these tools allow us to understand why a current deviation can still be corrected by future outputs, why success has become unlikely, or why the benchmark is already unattainable.
Remarkably, such information can also inform an external mitigation procedure. For example, if a completed collection of twelve images contains nine
male and three female engineers, the largest balanced selection contains
six images: three from each category. It can be obtained by removing
six male images and retaining the original order of the remaining
outputs. The resulting shorter series can then be evaluated against
the same representation requirement. The logic supplies the conditions
for diagnosing the original distribution and checking the revised one;
the selection and removal of outputs constitute a separate intervention.


The remainder of this paper is organized as follows. 
Section~\ref{sec:related} provides a detailed overview of related works, comparing our framework to existing counting and probabilistic logics.
In Section~\ref{sec:Preliminaries}, we recall for completeness the standard preliminary concepts of statistics and probability theory that form the mathematical background of our work. Section~\ref{sec:informal} presents informal yet concrete toy examples that our logic aims to formalize, alongside the intuitive meaning of the operators introduced to handle them.
In Section~\ref{sec:logic}, we formally establish the syntax and semantics of the  Temporal Logic of Frequencies on Finite Series ($\PL$).
Section~\ref{sec:properties} shows how to use $\PL$ to calculate key properties of event series, including the frequency of a specific outcome, the compatibility of an observed frequency with an ideal reference frequency, and the probability of a series terminating within such an ideal range.
Finally, Section~\ref{sec:conclusion} delineates future research directions.

\section{Related Works}\label{sec:related}

Our logic belongs to a broad family of formalisms in which modal and temporal logics are enriched with counting, weights, probabilities, durations, or long-run averages. Some of our modalities are also inspired by counting operators developed in connection with complexity theory~\cite{Wagner,ADLP} and finite model theory~\cite{Mostowski,Kontinen}. These approaches differ principally in the object measured by their quantitative component: the number of accessible worlds satisfying a formula, the number or relative frequency of satisfaction points along a trace, the duration of satisfaction within a time interval, or the probability measure assigned to a set of paths. In its present form, $\PL$ is best classified as a frequency-sensitive temporal logic over finite event series. Its distinctive purpose is to compare the empirical distribution observed along such a series with an independently specified target distribution and to determine, from an incomplete prefix, whether the target can still be attained.

Our framework can be presented as a conservative extension of Linear Temporal Logic over finite traces, $\LTLf$~\cite{degiacomo2013ltlf}. The language of $\LTLf$ contains the standard Boolean connectives and temporal operators of \textsf{LTL}, while formulas are interpreted over finite rather than infinite traces. The resulting semantics must therefore account explicitly for the final position of a trace, in particular when interpreting operators such as \emph{next} and \emph{until}. We preserve these operators and provide corresponding satisfaction clauses over our finite linear structures. On formulas belonging to the $\LTLf$ fragment, our semantics agrees with the ordinary finite-trace interpretation. The language of $\PL$ consequently retains the temporal expressiveness of $\LTLf$ while extending it with operators that express distributional constraints on a designated target random variable. This extension is conservative at the level of the standard temporal vocabulary, but introduces a different class of quantitative properties. Ordinary $\LTLf$ can describe the order in which events occur and the temporal relations among them, but it does not directly compare the relative frequency observed in a trace prefix with a benchmark distribution. Nor can it quantify over the possible completions of that prefix that satisfy the benchmark. The new operators of $\PL$ are designed to support precisely these verification and prediction tasks.

An extension of non-temporal modal logic with counting modalities is presented in~\cite{LEGASTELOIS2017}, where necessity and possibility are defined by counting how many accessible worlds satisfy a formula. Although the underlying intuition is related to ours, $\PL$ restricts the accessibility relation to a finite linear order interpreted temporally. More importantly, it includes modalities that relate the frequency observed in the actual series to a target frequency represented by a set of compatible assignments. Its semantics is therefore closer to that of quantitative temporal logics, although the presence of the target distribution gives it a different model-theoretic structure.

A first relevant family of temporal formalisms imposes absolute or relative counting constraints on traces. Counting \textsf{LTL}~\cite{Laroussinie2010} annotates the \emph{until} operator with Boolean combinations of linear constraints on the number of positions satisfying selected formulas before the witness for the right-hand side of \emph{until}. Thus, a formula $\phi U_{[C]}\psi$ requires a future position satisfying $\psi$, with $\phi$ holding at all intervening positions and the corresponding prefix satisfying the counting constraint $C$. Counting extensions have also been developed for branching-time logics, where the relevant constraints may concern states or paths in a computation tree~\cite{10.1007/978-3-642-12032-9_15}. Metric Temporal Logic with Counting extends the linear approach to timed words by introducing modalities that count occurrences within specified metric intervals and by allowing threshold constraints inside \emph{until}~\cite{KrishnaMadnaniPandya2016}. These systems constrain absolute numbers of occurrences, whereas $\PL$ directly evaluates relative frequencies over finite observed prefixes and relates them to a final target distribution.

Frequency \textsf{LTL} is defined in~\cite{Bolling2012} by replacing ordinary \emph{until} with a frequency-annotated operator. A formula $\phi U_c\psi$ holds when some future position satisfies $\psi$ and $\phi$ holds at no less than a $c$-fraction of the preceding positions. This construction is closer to $\PL$ than absolute-counting \textsf{LTL}, but the two systems differ in three respects. First, Frequency \textsf{LTL} evaluates a future segment ending at an existentially selected witness, whereas the frequency-sensitive verification operators of $\PL$ evaluate the prefix preceding the current point. Second, Frequency \textsf{LTL} is interpreted over infinite words, while our models represent finite series of predetermined length. Third, it contains no semantic counterpart of the set of assignments compatible with a target distribution. It therefore cannot express whether an observed partial frequency can still be completed to that distribution or support predictions conditional on such completions. The satisfiability problem for unrestricted Frequency \textsf{LTL} is undecidable, although a positive non-nested fragment is decidable and strictly more expressive than standard \textsf{LTL}; see~\cite{BolligDeckerLeucker2020Erratum} for further details.

The closest formalism to our treatment of temporal orientation is the finite-trace logic studied in~\cite{BednarczykMichaliszyn2021}. Its past-majority operator $PM\phi$ states that at least half of the positions preceding the current one satisfy $\phi$, while its more general percentage operator $P_{\bowtie k\%}\phi$ compares the proportion of such positions with an arbitrary threshold. Combining even weak past-frequency operators with future eventuality leads to undecidable satisfiability and model checking. Decidability is recovered for a restricted version of $\mathsf{LTL}^{\%}$ in which percentage operators are neither nested nor occur under negation. As in $\PL$, these operators look backwards along finite traces. Their thresholds, however, constrain only the observed prefix: the semantic structure contains no independently specified target distribution, space of compatible completions, or predictive probability.

Another family of temporal logics measures asymptotic frequencies over infinite executions. Mean-Payoff Temporal Logic associates rewards with the satisfaction of temporal subformulas and constrains the limit inferior or limit superior of their long-run mean payoff~\cite{TomitaEtAl2012MeanPayoff}. By using Boolean indicator rewards, it can express the long-run frequency with which a formula is satisfied. Forejt and Kr\v{c}\'{a}l extend \textsf{LTL} with an operator $G^{p}\phi$, which holds when the lower asymptotic density of suffixes satisfying $\phi$ is at least $p$~\cite{ForejtKrcal2015}. Their analysis is carried out over Markov chains and Markov decision processes. Subsequent work establishes the decidability of controller synthesis for the fragment in which \emph{until} does not occur within ordinary or frequency-global operators~\cite{ForejtKrcalKretinsky2015}. Average-\textsf{LTL} similarly introduces long-run and average-until modalities, but assigns quantitative values to formulas rather than testing Boolean frequency thresholds; computing or even approximating these values is undecidable~\cite{BouyerMarkeyMatteplackel2014}. All these formalisms concern limiting behaviour over infinite traces. By contrast, $\PL$ evaluates exact empirical frequencies at intermediate points of a finite series and relates them to a fixed final distribution.

Duration-based and signal logics provide a further point of comparison. Duration Calculus reasons about the accumulated amount of continuous time for which a state holds within an interval~\cite{ZhouHoareRavn1991}; a relative-frequency constraint can be obtained by comparing this duration with the length of the interval. Averaged Signal Temporal Logic introduces temporal operators whose robustness values are averaged over metric intervals~\cite{AkazakiHasuo2015}. Signal Convolution Logic uses convolution kernels to express the percentage of time during which a formula is satisfied in a bounded interval and provides both Boolean and quantitative semantics~\cite{SilvettiEtAl2018}. More recently, Cumulative-Time Signal Temporal Logic has introduced a dedicated operator $\mathbf{C}^{\tau}_{I}\phi$, requiring that $\phi$ hold for at least $\tau$ time units or samples within the window $I$~\cite{ChenEtAl2025CTSTL}. For a fixed interval, setting $\tau$ to a proportion of its length yields a frequency condition. These logics are designed primarily for dense-time or sampled signals and bounded-window monitoring. They do not reason about the possible completions of a finite experimental series relative to a desired distribution.

Probabilistic temporal logics quantify a different object. \textsf{PCTL} attaches probability bounds to sets of paths satisfying formulas constructed using operators such as \emph{next} and \emph{until}~\cite{hansson1994logic}. The resulting value is a probability over alternative executions, rather than the frequency of satisfaction points within a single execution. Related systems include variants of \textsf{PCTL} and \textsf{CSL}~\cite{Aziz1995,Aziz1996,Kwiatkowska2001}, epistemic \textsf{PCTLK}~\cite{WAN2013279}, and Computationally-Grounded Weighted Doxastic Logic~\cite{DBLP:journals/sLogica/ChenPRR16}. These formalisms respectively express probabilities of temporal behaviours or degrees of belief concerning them. In $\PL$, by contrast, uncertainty is evaluated by an external observer over the possible completions of a finite series; agents and their epistemic states are not represented in the language.

Probabilistic Frequency Temporal Logic makes the distinction between path probability and within-path frequency explicit. It combines a $\textsf{CTL}^{*}$-style path-probability operator with a conditional-frequency operator $\mathbf{Q}^{I}_{\bowtie q}\langle\phi\mid\psi\rangle$, which measures the proportion of time points in $I$ satisfying $\phi$ among those satisfying $\psi$~\cite{Tomita2011PFTL}. The logic supports discrete and continuous time as well as finite and unbounded intervals. Although this provides a genuine within-path frequency operator, its interpretation presupposes a stochastic transition model. By contrast, the compatible assignments in $\PL$ represent all finite series realizing a stipulated target frequency. Probability is then used to evaluate possible continuations relative to this explicitly represented benchmark. A closely related approach is the probabilistic temporal logic over finite traces based on superposition semantics~\cite{MaMP25}. It represents uncertainty by treating alternative temporal evolutions as simultaneously possible until observations exclude them, and supports probabilistic inference through automata-based methods. Like $\PL$, this framework reasons over incomplete finite traces and their possible continuations. Its objective is nevertheless different: probability concerns alternative evolutions compatible with temporal specifications, whereas $\PL$ evaluates the empirical frequency of a designated outcome against an independently specified benchmark and uses compatible completions to determine whether that benchmark can still be attained.

The explicit use of a target frequency connects our proposal with the semantics developed in~\cite{DBLP:journals/ijar/KubyshkinaP24}, where probabilistic and frequency statements describe idealized and empirical processes and their outputs. That approach compares an empirical process with a theoretical probabilistic model. This semantics is sound with respect to a natural deduction calculus of experimental series as complete and static outcomes~\cite{DasaroGencoPrimiero,ceragioli_primiero}. In $\PL$, this comparison is internalized within a single temporal structure: the actual assignment records the observed series, while $F_{\Omega}$ collects the assignments compatible with the target distribution. The temporal perspective additionally supports evaluation at every stage of an incomplete series.

The resulting modalities can express not only the frequency already observed, but also its agreement with the benchmark, the possibility or probability of reaching the benchmark from a partial observation, and predictions concerning subsequent outcomes. The observed frequency is evaluated retrospectively from the finite prefix available at the current point, while the predictive operators relate that prefix to the remaining length of the series and to the target distribution. This combination of finite-trace temporal reasoning, past-directed frequency evaluation, an internal representation of the target distribution, and prediction over compatible completions distinguishes $\PL$ from existing counting, frequency, mean-payoff, signal, and probabilistic temporal logics.

The present work therefore contributes to a broader research programme devoted to the logical formalization of experimental series~\cite{DasaroGencoPrimiero,DBLP:journals/ijar/KubyshkinaP24,ceragioli_primiero}, while shifting its focus from completed outcomes to dynamic verification and prediction. We restrict the analysis here to finite linear structures. An extension to tree-based models is under development to consider branching continuations and more complex predictive scenarios. The construction of a proof system and the investigation of further meta-theoretical properties are left for future work; see Section~\ref{sec:conclusion}.

\section{Preliminaries}\label{sec:Preliminaries}
For the sake of self-containment, in this section we recall standard notions in statistics and probability theory and make our notational conventions explicit. For further details, see~\cite{Billingsley}.

\bigskip
An \emph{experiment} or \emph{trial} is the mathematical model of a procedure that can be repeated. 
%
Each experiment is associated with a set of possible values taken as its single outcome; e.g.~for the toss of a coin, the set of possible values or single outcomes is defined by the set $\{\head, \tail\}$, while for a die it is $\{1,\dots, 6\}$.
The repeated trials of an experiment are called a sample or series.
It is associated with a set of possible multiple outcomes, which depends on the possible single outcomes and on the number of times the experiment is repeated.
%
%
As standard, we use $\omega$ to denote an outcome or point of a sample, i.e.~the result of a single trial of an experiment, and $\Omega$ to denote the sample space, i.e.~the set of all possible outcomes $\omega$'s associated with this sample.
Recall that outcomes must be \emph{mutually exclusive}, i.e.~if $\omega_i$ occurs, for any $j\neq i$, no other $\omega_j$ can take place, and \emph{collectively exhaustive}, i.e.~one of the outcomes will result in every trial.

\begin{ex}
\label{ex:outcome}
\begin{sloppypar}
    Let us consider a sample of $n$ coin tosses.
    Each possible outcome is a sequence of $n$ single outcomes resulting from tossing, i.e.~$\omega=\omega(1)\omega(2)\dots \omega(n)$, where $\omega(i)\in \{\head,\tail\}$ for any $i\in \{1,\dots, n\}$, while $\Omega$ is the set of all $2^n$ possible outcomes. 
\end{sloppypar}
\end{ex}
\noindent
Subsets of $\Omega$ are called events, $E_1, E_2, \dots$.
An event consisting of only one outcome is called an \emph{atomic} or \emph{elementary event}, while
an event that is made up of more than one possible outcome is called a \emph{compound event}. 
Two events, say $E_1$ and $E_2$ are \emph{disjoint} or \emph{mutually exclusive} when they cannot happen at the same time, i.e.~$E_1 \cap E_2=\varnothing$.

To define measurable sets, we rely on the standard notion of $\sigma$-algebra.

\begin{definition}[Algebra]
    A class $\mathcal{A}$ of subsets of $\Omega$ is called an \emph{algebra} if it contains $\Omega$ itself and is closed under complementation and finite unions:
       i. $\Omega \in \mathcal{A}$;
       ii. if $E\in \mathcal{A}$, then $\overline{E} \in \mathcal{A}$, where $\overline{E}$ denotes the complement of $E$;
        iii. if $E_1, E_2 \in \mathcal{A}$, then $E_1 \cup E_2 \in \mathcal{A}$.
    An algebra is closed under finite set-theoretic operations, while a $\sigma$-algebra is closed also under countable ones.
\end{definition}

\noindent
The first condition simply ensures that $\mathcal{A}$ is nonempty.
By De Morgan's law (and (ii.)), condition (iii.) could be replaced by the alternative condition (iii.') saying that if $E_1, E_2 \in \mathcal{A}$, then $E_1 \cap E_2 \in \mathcal{A}$.
In probability theory, we are particularly interested in the smallest $\sigma$-algebras containing a given algebra $\mathcal{A}$, namely $\sigma(\mathcal{A})$, defined below.

\begin{ex}
    The largest $\sigma$-algebra on $\Omega$ is the power set $2^\Omega$, while the smallest one is $\{\varnothing, \Omega\}$
\end{ex}

\noindent
The $\sigma$-algebra generated by an algebra is the smallest $\sigma$-algebra containing it.

\begin{definition}[Generated $\sigma$-algebra]
   The \emph{$\sigma$-algebra generated by $\mathcal{A}$, $\sigma(\mathcal{A})$}, is the intersection of all (and only) the $\sigma$-algebra containing $\mathcal{A}$:
       i. $\mathcal{A} \subseteq \sigma(\mathcal{A})$
      ii. $\sigma(\mathcal{A})$ is a $\sigma$-algebra
      iii. if $\mathcal{A}\subseteq \mathcal{A}'$ and $\mathcal{A}'$ is a $\sigma$-algebra, then $\sigma(\mathcal{A})\subseteq \mathcal{A}'$.
\end{definition}

\noindent
A \emph{measurable space} is a pair $(\Omega, \mathcal{A})$, where $\mathcal{A}$ is a $\sigma$-algebra over $\Omega$. Finally, a \emph{probability measure}, $\pr(\cdot)$, on a $\sigma$-algebra $\mathcal{A}$ is a function associating any event $E\in \mathcal{A}$ with a number $\pr(E)$ so that:
    i. for each $E\in \mathcal{A}$, $0\leq \pr(E) \leq 1$
    ii. $\pr(\varnothing)=0$ and $\pr(\Omega)=1$
    iii. if $(E_k)^{\infty}_{k=1}$ is a sequence of pairwise disjoint events in $\mathcal{A}$, then
    $ \pr\Big(\bigcup^{\infty}_{k=1} E_k \Big) = \sum^{\infty}_{k=1} \pr(E_k).
    $

\noindent
Two events are \emph{(stochastically) independent} when the occurrence of one does not affect the probability for the other to occur.
In particular, given two disjoint events, $E_1$ and $E_2$, $\pr(E_1\cup E_2)=\pr(E_1)+\pr(E_2)$, while for two independent events, $E_1'$ and $E_2'$, $\pr(E_1'\cap E_2')=\pr(E_1') \cdot \pr(E_2')$.
Then, a probability space is a formal model for random process obtained by endowing a measurable space with a corresponding probability measure.

\begin{definition}[Probability space]\label{def:prob}
    A \emph{probability space} $(\Omega, \mathcal{A}, \pr)$ is made of:
    \begin{itemize}
        \itemsep0em
        \item a \emph{sample space} $\Omega$, which is the set of all possible outcomes;
        \item a $\sigma$-algebra $\mathcal{A}$, which is the set of events;
        \item a \emph{probability measure} $\pr: \mathcal{A} \to [0,1]$, which assigns to each event in $\mathcal{A}$ a probability, i.e.~a number between 0 and 1 satisfying the so-called Kolmogorov axioms.
    \end{itemize}
\end{definition}

Since our logic is meant to describe actual experiments and empirical probability, we focus on the notions of frequency and time series.

\begin{definition}[Frequency and relative frequency]
    Given a series of $n$ trials of an experiment, the \emph{absolute frequency} of a single outcome or value is the number of trials, among the $n$ in the series, that result in this given value.
    The \emph{relative frequency} or \emph{empirical probability} of a value is the ratio of the number of times in which the specified value occurs to the total number of trials.
\end{definition}

\begin{notation}
    Given an experiment with single outcomes in $\{p_1,\dots,p_m\}$ and $n$ trials of this experiment, for any $i\in \{1,\dots,m\},$ we use $\fr(p_i) \in \mathbb{Q}_{[0,1]}$ to denote the relative frequency of the value $p_i$.
\end{notation}

\begin{ex}
    Let us assume we flip a coin 100 times and we obtain head 50 times. Then the frequency of obtaining head in our experiment is 50 and its relative frequency is $\frac{1}{2}$, i.e.~$\fr(\head)=\frac{1}{2}$.
\end{ex}

\begin{definition}[Probability of Frequencies]
    Given a series of $n$ independent trials of an experiment whose single outcomes are equiprobable, the probability of obtaining a particular (absolute or relative) frequency is given by the ratio of the number of possible series compatible with that frequency to the number of all possible series (see Theorem~\ref{conj:probabilityFrequence}). The map sending each frequency distribution to this value is a probability mass function on the compound events grouping series with the same frequency distribution. The general case, in which single outcomes need not be equiprobable, is treated in Theorem~\ref{conj:probabilityFrequenceGeneral}.
\end{definition}

\section{Reasoning about series of events}\label{sec:informal}

In this section, we informally outline the core ideas behind our logic, which naturally expresses properties of finite series of events and reasons about the frequency of specific outputs.
To ground these concepts, we introduce a running toy example: Section~\ref{sec:guide} presents the example informally, while Section~\ref{sec:keyOp} details the core operators of our logic.

\subsection{A guiding example}\label{sec:guide}

As anticipated, we deal with series consisting of a pre-determined number of trials.
In our framework, it is natural to use atomic propositions to represent \emph{single} outcomes of such series.
We use $\mu_{\mathbf{a}}(p)$ to denote the relative frequency of $p$, as introduced in Section~\ref{sec:Preliminaries}, while we use $\fr(p)$ to denote a \emph{target} or \emph{ideal} frequency for $p$, i.e.~a frequency distribution we want to compare the series against, as opposed to the empirical frequency
actually observed. Moreover, we use \emph{series} to refer to the abstract sequence of $n$ trials, and \emph{realization} to refer to one concrete instance of the series, i.e.~a specific sequence of outcomes obtained by actually performing the $n$ trials.
For example, given a series of trials, we want to formally study the probability that the frequency of the outcomes mirrors some expected or desired frequency.
We investigate the properties of events, both assuming that all information about the series is given, and (more interestingly) under partial knowledge in which we know only the results of a subset of all trials in the series; typically, the values of the first \(m\) ones.
In the latter scenario, we also investigate the probability of the outcomes of the subsequent event(s) under the assumption that the frequency of outputs of all the experiments is some ``ideal frequency''. We now formulate some problems to guide our analysis in the rest of this paper.

\bigskip
Let us assume that our sample consists of \(n\) coin tosses with \(\Omega=\{\head\ (H), \tail\ (T)\}\). 
We start from a simple checking problem on frequencies:
\begin{restatable}{thm}{probwhite}\label{probwhite}
Given a series of \(m \leq n\) events of the series, e.g., the mentioned coin tosses, can we check, at any given time, whether the frequency of any given outcome is greater than or equal to \(q\)?
\end{restatable}
\begin{ex}\label{exwhite}
Consider a realization with results \(HHT\).
Clearly, at the third throw it holds that the relative frequency of \(H\) is \(\frac{2}{3}\).
\end{ex}
\noindent
It could be more interesting to relate and compare the current results with a desired or target frequency.

\begin{restatable}{thm}{probblack}\label{probblack}
Given a target frequency (e.g.,~for coin tosses), can we check which partial frequencies events are possible at any given time?
\end{restatable}
\begin{ex}\label{exblack}
Consider four coin tosses such that the frequency of obtaining \(H\) came out to be \(\frac{1}{2}\). 
Clearly, at the third toss, the observed (partial) frequency of \(H\) could be either \(\frac{1}{3}\) or \(\frac{2}{3}\).
\end{ex}

\noindent
A related decision problem concerning frequencies is presented below, which checks the compatibility of the results obtained so far with a target frequency.

\begin{restatable}{thm}{probOne}\label{prob1}
Given a series of \(m\leq n\) executions of an experiment, e.g., of coin tosses, can we check whether the current observation is compatible with a given frequency?
\end{restatable}
\begin{ex}\label{ex1}
Consider four coin tosses and a target frequency \(\fr(\head)=\frac{1}{2}\): a series of outputs \(TT\) for the first two experiments is still compatible with the desired frequency; given a series of outputs \(TTT\) for the first three executions, we know that we will not be able to reach the desired frequency.
\end{ex}

\noindent
Next, this query can be refined into a quantitative problem.  
    
\begin{restatable}{thm}{probTwo}\label{prob2}
If compatible, given a series of \(m\leq n\) executions of the experiment, what is the probability of the given frequency to be realized?
\end{restatable}

\begin{ex}\label{ex2}
Given \(2^{4}\) (\(=16\)) equally probable series of single outputs and only \(\binom{4}{2}\) (\(=6\)) of them satisfying the desired frequency \(\fr(\head)=\frac{1}{2}\), the probability of reaching it before any execution is \(\frac{6}{16}=\frac{3}{8}\).\footnote{Notice that, under the assumption that \(\fr(\head)=\fr(\tail)\), assignments are equiprobable with each other.}
Similarly, given a series of outputs \(TT\) for the first two executions, the probability to end with the desired frequency \(\fr(\head)=\frac{1}{2}\) is \(\frac{1}{4}\).
\end{ex}

\noindent
Under the assumption that the given frequency is realized, a natural question concerns the behavior of the series at any next step.

\begin{restatable}{thm}{probThree}\label{prob3}    
Given a series of \(n\) experiments for which a given frequency of outputs is realized, what is the probability at \(m<n\) to obtain a certain output at \(m+i\leq n\)?
\end{restatable}

\begin{ex}\label{ex3}
Given the series of outputs \(TT\) for the first two experiments, the probability that the third experiment will give \(H\) is $1$, under the assumption that the desired frequency (the same of Example~\ref{ex2}) is realized.
\end{ex}

\subsection{Key operators, informally}\label{sec:keyOp}

The logic we are going to introduce in Section~\ref{sec:logic} allows us to formalize all the problems presented in the previous section in an extremely natural way, relying on some original modal operators.
The three operators defined below allow, in different forms, to count the number of times an event has occurred so far or can occur according to an expected or known frequency \(\Pg\) associated with our realization.

\paragraph{The operator $\wbox$.}
The operator $\wbox$ is a world- and realization-dependent operator, which intuitively expresses the fact that, until a given time \(t_w\), the single outcome \(p\) has been observed \emph{at least} \(q\in \mathbb{Q}_{[0,1]}\) of the times.
This operator allows us to describe the relative frequency of a single outcome at the end of the time series or for any sub-series obtained by ``cutting'' the series at a given time.

\begin{figure}[!ht]
    \centering
    \resizebox{5cm}{!}{%
        \begin{circuitikz}
            \tikzstyle{every node}=[font=\large]
           \node at (7.5,12.25) {$t_{1}$};
            \draw (7.5,11.25) circle (0.5cm) node {H} ;
            \draw (9.25,11.25) circle (0.5cm) node {H} ;
            \node at (9.25,12.25) {$t_{2}$};
            \node at (11,12.25) {$t_{3}$};
            \draw (11,11.25) circle (0.5cm) node {T} ;
            \node at (12.75,12.25) {$t_{4}$};
            \draw (12.75,11.25) circle (0.5cm) node {H} ;
            \end{circuitikz}
        }%
    \caption{A model representing four coin tosses resulting in $HHTH$}.
    \label{fig:wbox}
\end{figure}
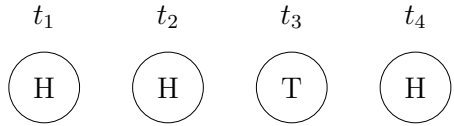

\begin{ex}\label{ex:wbox}
    The model in Figure~\ref{fig:wbox} represents the observation of four coin tosses resulting in $HHTH$. Thanks to the $\wbox$ operator we can express that after the second toss of the coin, i.e.~at time $t_2$, the relative frequency of observing $\head$ is $1$, $\wbox_{1}(\head)$. 
    The empirical probability of observing $\head$ changes to $\frac{2}{3}$ after the third toss, i.e.~at $t_3$, $\wbox_{\frac{2}{3}}\head$ holds, and then to $\frac{3}{4}$ at the end of the series, i.e.~the formula $\wbox_{\frac{3}{4}}\head$ holds at $t_4$.
\end{ex}

\paragraph{The operator $\bbox$.} The operator $\bbox$ is again a time-dependent operator, but, differently from $\wbox$, it is observation-independent. 
Intuitively, it expresses that, according to a given or target frequency $\Pg$, up until a given time $t_m$, the value of interest (e.g.,~$H$) could have happened \emph{at least} $q$ of the times.

\begin{figure}[!ht]
    \centering
    \resizebox{7cm}{!}{%
        \begin{circuitikz}
            \tikzstyle{every node}=[font=\large]
            \node at (7.5,12.25) {$t_{1}$};
            \draw (7.5,11.25) circle (0.5cm) node {H} ;
            \draw (9.25,11.25) circle (0.5cm) node {H} ;
            \node at (9.25,12.25) {$t_{2}$};
            \node at (11,12.25) {$t_{3}$};
            \draw (11,11.25) circle (0.5cm) node {H} ;
            \node at (12.75,12.25) {$t_{4}$};
            \draw (12.75,11.25) circle (0.5cm) node {T} ;
            \node at (14.5,12.25) {$t_{5}$};
            \draw (14.5,11.25) circle (0.5cm) node {T} ;
            \node at (16.25,12.25) {$t_{6}$};
            \draw (16.25,11.25) circle (0.5cm) node {T} ;
            \end{circuitikz}
        }%
    \caption{A model representing six coin tosses resulting in $HHHTTT$}.
    \label{fig:bbox}
\end{figure}
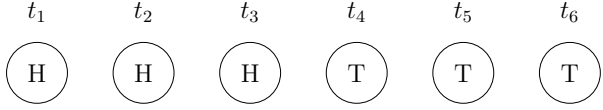

\begin{ex}
    Let us assume that we think our coin is (relatively) unbiased, i.e.,~that the distribution associated with our series is \(\fr(\head)=\frac{1}{2}\) and \(\fr(\tail)=\frac{1}{2}\). 
    Given any series of \(n\ge2\) tosses, any outcome for the first \(\frac{n}{2}\) tosses (for example, obtaining only heads (\(\head\)) or only tails (\(\tail\))) is compatible with the desired distribution \(\fr\). 
    This is expressed by the fact that the formulas \(\bbox_q \head\) and \(\bbox_q \tail\) are true for every \(q \in \mathbb{Q}_{[0,1]}\) (in particular, for \(q=1\)) at any time \(t_m\) with \(m\leq \frac{n}{2}\); after passing the first half of the series, this is not true anymore.
    On the other hand, at the end of the series (i.e.,~considering the last time \(t_n\)), a situation is compatible with the desired frequency \(\fr\) if half of the tosses gave \(\head\) and the other half gave \(\tail\), which is expressed by the fact that \(\bbox_{\frac{1}{2}}\head\) and \(\bbox_{\frac{1}{2}}\tail\) are satisfied at time \(t_n\) (actually holding for any \(q\leq \frac{1}{2}\)).
    Therefore, if we consider the model in Figure~\ref{fig:bbox}, representing six coin tosses resulting in the unbiased behavior \(HHHTTT\), \(\bbox_1 \head\) is satisfied for any time up until \(t_3\), whereas at the end of the series, namely at time \(t_6\), only the weaker formula \(\bbox_{\frac{1}{2}} \head\) is satisfied.
\end{ex}

\paragraph{The operator \(\sat\).} As the operator \(\wbox\), the operator \(\sat\) is time- and realization-dependent. 
A circled formula \(\sat_{q}\fOne\), holding at a time \(t_m\), expresses that until that time the outcome corresponding to \(\fOne\) has been obtained at least \(q \times n\) times, where \(n\) is the total number of trials executed in the considered realization series.
In other words, it can be seen as the lowest attainable frequency of \(\fOne\) at the conclusion of the series, given that up until time \(t_m\), we know how many times the event of interest has occurred.

\begin{ex}
    Going back to Example~\ref{ex:wbox}, let us assume that, again, we have flipped the coin \(m\) times, \(l\) times the outcome has been head (\(\head\)), and that we know that we will flip the coin \(n\) times in total. 
    Then, we can express that, at the end of the experiment, we will obtain head with \emph{at least} probability \(\frac{l}{n}\) by saying that, according to the current realization, the formula \(\sat_{\frac{l}{n}}\head\) is satisfied at time \(t_m\). 
    Considering again the model in Figure~\ref{fig:wbox}, we can now say that at time \(t_1\) the formula \(\sat_{\frac{1}{4}}\head\) is satisfied; at \(t_2\) and \(t_3\) \(\sat_{\frac{1}{2}}\head\) holds; at \(t_4\) the formula \(\sat_{\frac{3}{4}}\head\) holds.
    Intuitively, this operator provides information about the number of times a single outcome has come out (in the concrete case considered, and at a specific time, \(m\)), taking into account the total number \(n\) of trials constituting the series.
\end{ex}

\paragraph{The next operator \(\wnext\).} The next operator \(\wnext\) is time-dependent. It expresses that, assuming that we know that the current realization is associated with a given distribution, the probability that event \(\fOne\) happens at the next time \(t_{m+1}\) is at least \(q\).
This operator can be generalized to \(\wnext^i\), expressing the probability that an event happens at least once in the next \(i\) executions of the experiment.

\begin{figure}[!ht]
    \centering
    \resizebox{5cm}{!}{%
        \begin{circuitikz}
            \tikzstyle{every node}=[font=\large]
            \node at (7.5,12.25) {$t_{1}$};
            \draw (7.5,11.25) circle (0.5cm) node {T} ;
            \draw (9.25,11.25) circle (0.5cm) node {T} ;
            \node at (9.25,12.25) {$t_{2}$};
            \node at (11,12.25) {$t_{3}$};
            \draw (11,11.25) circle (0.5cm) node {H} ;
            \node at (12.75,12.25) {$t_{4}$};
            \draw (12.75,11.25) circle (0.5cm) node {H} ;
            \end{circuitikz}
        }%
    \caption{A model representing four coin tosses resulting in $TTHH$}.
    \label{fig:next}
\end{figure}
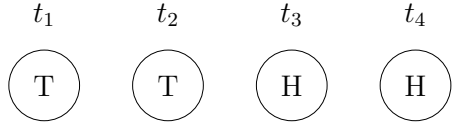

\begin{ex}\label{ex:next}
    Consider the model of Figure~\ref{fig:next}. 
    Given an unbiased coin and a series of four coin tosses, if we know that the first two tosses return tail (\(\tail\)), the fact that at the third toss the outcome is (certainly) head (\(\head\)) is expressed by the formula \(\wnext_1 \head\) being satisfied at the time corresponding to the second toss \(t_2\).
    As we shall see, the facts that the series is made of four tosses and that the coin is (relatively) unbiased are part of the definition of our model, while the fact that the first two tosses end up with \(\tail\) depends on the given realization (see Section~\ref{sec:semantics}).
\end{ex}

\begin{ex}
    Considering again the model of Figure~\ref{fig:next}, the fact that the probability that either the third or the fourth coin toss returns \(\head\) is (at least) $1$ is expressed by the formula \(\wnext^2_1 \head\) being satisfied at time \(t_2\) corresponding to the second toss.
\end{ex}

\paragraph{Global star operator \(\bstar\).}
Finally, a star formula \(\bstar_{q}\fOne\) expresses that, according to a given probability distribution associated with our series of experiments, the outcome \(\fOne\) has a global probability (or relative frequency) of at least \(q\).
Notably, ``stars are blind'' in the sense that this operator is independent of both the actual experimental realization and the current time we are in (i.e.,~\(t\)). 

\begin{ex}
    If the coin is (relatively) unbiased, for any realization and at any time we have a probability of \(\frac{1}{2}\) (thus also greater \emph{or equal} to \(\frac{1}{2}\)) of obtaining head (\(\head\)).
    This is expressed by the satisfaction of the formula \(\bstar_{\frac{1}{2}}\head\) at every time \(t_m\).
    Similarly, if the coin is (relatively) biased and returns \(\head\) only \(\frac{1}{3}\) of the times, then clearly for any time of any actual realization the star formula \(\bstar_{q} \head\) is satisfied when \(q\leq \frac{1}{3}\).
\end{ex}

\section{The temporal logic of frequencies on series $\PL$}\label{sec:logic}

In this section, we introduce the formal language and semantics of $\PL$, a natural formalism for expressing properties of experiments occurring in a series or for comparing the frequency of given phenomena with an expected distribution.

\subsection{The language of $\PL$}
The language of $\PL$ is defined as a standard language for linear temporal logic over finite traces \LTLf\ extended with modal operators corresponding to those informally presented in Section~\ref{sec:keyOp}.

\begin{notation}
    In what follows, we use $\at_1, \at_2, \dots$ for atomic propositions and $\fOne,\fTwo,\fThree,\dots$ for compound ones. 
    We use $\At$ to denote the set of atomic propositions, and $\At_\Omega \subseteq \At$ for any finite subset corresponding to a sample space $\Omega$.
\end{notation}

\begin{definition}[Language of $\PL$]
    Formulas of $\PL$ are defined by the grammar below:
    $$
    \fOne ::= \at \; | \; \neg \fOne \; | \; \fOne \wedge \fOne \; | \; \fOne \vee \fOne \;  | \; \bigcirc \phi \; | \; \phi \ \Until \ \psi \; | \; \wbox_q \fOne \; | \;  \bbox_q \fOne\; | \; \sat_q \fOne \; | \; \wnext_q \fOne \; | \; \bstar_q \fOne
    $$
    where $\at \in \At$ and $q\in \mathbb{Q}_{[0,1]}$.
\end{definition}

\noindent
Notably, our language allows for nested modalities, including interactions with temporal operators \emph{next} and \emph{until}.

\begin{ex}
    Consider again the situation depicted in Figure~\ref{fig:wbox}.
    As seen, the formula $\wbox_1 \head$ intuitively expresses that, up to the evaluation time, (the outcome corresponding to) $\head$ has occurred all the times.
    Clearly, this formula holds both at $t_1$ and at $t_2$.
    Moreover, since the relative frequency of $\wbox_1 \head$ in the tosses previous to $t_2$ is $1$, also the nested formula $\wbox _1 \wbox _1 \head$ holds at both worlds $t_1$ and $t_2$.
    Moreover, notice that only half of the tosses preceding $t_4$ satisfy $\wbox_1 \head$.
    Hence, while $\wbox_{\frac{3}{4}}\head$ holds at $t_4$, only $\wbox_{\frac{1}{2}} \wbox_1 \head$ holds there.
%
\end{ex}

\begin{ex}
    Let's consider again the situation in Figure \ref{fig:wbox}. The standard \LTLf\ formula $\bigcirc\head$ simply expresses that, at the next time point in the series, $\head$ holds. Let's consider now its interaction with one of our modal operators. The formula $\bigcirc\wbox_{1}\head$ intuitively expresses that, \emph{at the next time point in the series}, $\head$ will have occurred all the times. Then, this formula only holds at $t_1$, since at the next time step $t_2$ $\head$ has in fact occurred all the times. 
\end{ex}

\noindent
For simplicity, we usually restrict our attention to sets of atoms corresponding to possible outcomes over a finite sequence of trials. This convention will be assumed throughout the remainder of the paper. Formally, this is established via a correspondence between outcomes in a given trial and atomic formulas in our language.

\begin{definition}[Outcome-atom correspondence]
\label{def:outcome-atom-correspondence}
Given a finite discrete probability space with sample space $\Omega$, let $\At_\Omega\subseteq\At$ be a finite set of mutually exclusive and
collectively exhaustive atomic formulas such that
$|\At_\Omega|=|\Omega|$. An \emph{outcome-atom correspondence} is a bijection
$\iota:\At_\Omega\longrightarrow\Omega$. 
\end{definition}


\subsection{The semantics of $\PL$}\label{sec:semantics}

We introduce a Kripke-style semantics for $\PL$ defined over a finite series of worlds.
Specifically, the corresponding structure consists of a finite set of worlds representing time instants, a special accessibility relation representing a linear timeline, and a set of assignments corresponding to a well-defined frequency distribution $\Pg$.
These worlds model the specific time instants at which events occur.
Notice that, although the semantics is standard in the sense that formulas are  either true or false in a given model and interpretation, boxed formulas are evaluated considering \emph{sets} of worlds and assignments.\footnote{In $\PL$, this is allowed directly by the modal operator `in one step', but alternative presentations are possible, e.g.~considering terms, the interpretation of which is a measurable set, see e.g.~\cite{FHM}.}

\begin{definition}\label{def:Assignments}
Given a finite set of atomic propositions $\At_\Omega \subseteq \At$ and a set of worlds $W$, an (atomic) assignment $a$ is a mapping $\At_\Omega \longrightarrow \mathcal{P}(W)$ such that for any $w\in W$ there is a unique $\at \in \At_\Omega$ for which $w \in a(\at)$; that is, $\{a(\at)\}_{\at \in \At_\Omega}$ is a partition of $W$, and for every $p \in \At \setminus \At_\Omega$ it holds $a(\at) = \varnothing$.
With abuse of notation, for $w \in W$ we write $a(w)$  to denote the unique atom $\at$ such that $w \in a(\at)$; while, then, $a$ acts both on atoms, returning a set of worlds, and on worlds, returning an atom, the intended meaning is always clear from the type of the argument.
\end{definition}

\begin{definition}
\label{def:PrL}
A \(\PL\)-structure is a tuple
$$
\str=(W,\Rexp,\fr,\Pg),
$$
where:
\begin{itemize}
\itemsep0em
    \item $W=\{w_1,\ldots,w_n\}$, with $n\geq 1$, is a finite set
    of worlds;
    \item $\Rexp$ is a finite linear (i.e.,~total) order on $W$;
    \item $\fr$ is a well-defined empirical probability measure over
    the $\sigma$-algebra of a probability space $\Omega$ derived from the outcome-atom correspondence bijection $\iota:\At_\Omega\longrightarrow \Omega$, and such that for every $\at\in\At_\Omega$, $n\fr(\at)\in\mathbb N$, and $\sum_{p\in \At_\Omega}\mu_\Phi(p)=1$;
    \item $\Pg\neq\varnothing$ is the set of all assignments
    $ a:\At_{\Omega}\longrightarrow \mathcal P(W)$ compatible with $\fr$, namely,
    those satisfying, for every 
    $\at\in\At_\Omega$,
    $$
    \frac{\bigl|\{w\in W:w\in a(\at)\}\bigr|}{|W|}
    =\fr(\at).
    $$
\end{itemize}
The integrality condition on \(\fr\) guarantees that \(\Pg\) is
non-empty. When an assignment
\(\vseg:\At\to\mathcal P(W)\), intuitively corresponding to the
``real'' observation, is added, the resulting tuple
$$
\mathbb M=(W,\Rexp,\fr,\Pg,\vseg)
$$
is called a $\PL$-model.
\end{definition}

\noindent
In fact, $\Pg$ and $\fr$ are interdefinable, meaning we could have defined our structure $\str$ more parsimoniously by omitting one of the two. However, since we frequently use both $\Pg$ and $\fr$ in our lemmas and explanations, we have chosen to make both explicit.

The relation $\Rexp$ is a (finite) linear order on $W$, i.e.,~it is reflexive, transitive, antisymmetric, and total: since it is reflexive, for any $w\in W$, $wR_ew$; since it is transitive, for any $w,w',w''\in W$, if $wR_ew'$ and $w'R_ew''$, then $wR_ew''$; since it is antisymmetric, there is no $w,w'\in W$ with $w\neq w'$, such that $wR_ew'$ and $w'R_ew$; since it is total, for any $w,w'\in W$, if $w\neq w'$, then either $wR_ew'$ or $w'R_ew$. Since $W$ is finite, we fix the enumeration $W=\{w_1,\dots,w_n\}$ so that it agrees with this order, i.e.\ $w_iR_ew_j$ iff $i\ge j$, for all $i,j\in\{1,\dots,n\}$; in particular, $\{w\in W : w_mR_ew\}=\{w_1,\dots,w_m\}$ for every $m\in\{1,\dots,n\}$.

\begin{definition}\label{def:successor}
For $w=w_i$ with $i<n$, the \emph{successor} of $w$, denoted $succ(w)$, is the world $w_{i+1}$, i.e.\ the unique $w'\in W$ such that $w \neq w', w' \Rexp w$ and there is no $w''\in W \text{ with } w \ne w'' \ne w'$ such that $w' \Rexp w'' \Rexp w$.
The world $w_n$ has no successor: $succ$ is a \emph{partial} function on $W$, defined exactly on $W\setminus\{w_n\}$. This also allows us to avoid defining a $last$ abbreviation (as standard in \LTLf) since the last state of the series is the world $w_n$.
\end{definition}

\begin{notation}
    For the sake of readability, in what follows we will assume that worlds in $W$ are labelled as $w_1, w_2, \dots, w_n$ and that $\Rexp$ is defined accordingly.
\end{notation}

It is now possible to define the satisfiability relation in a world of $\model$.

\begin{definition}\label{def:semantics}
    For any $\PL$-formula $\fOne$, $\model = (W, 
    \Rexp, \fr, \Pg , \vseg)$, such that $|W|=n$, assignment $a \in \Pg \cup \{\vseg \}$, and world $w_m\in W$,  we define the relation $\model, w_m \vDash_{a}$, inductively as follows:
   %
\begin{align*}
&\model, w_m \vDash_a p 
&&\quad \text{iff} \quad \quad w_m \in a(p) &\\
&\model, w_m \vDash_a \neg \fOne 
&&\quad \text{iff} \quad \quad \model, w_m \not\vDash_a \fOne &\\
&\model, w_m \vDash_{a} \fOne \wedge \fTwo 
&&\quad \text{iff} \quad \quad \model, w_m \vDash_{a} \fOne \text{ and } \model, w_m \vDash_{a} \fTwo &\\
&\model, w_m \vDash_{a} \fOne \vee \fTwo 
&&\quad \text{iff} \quad \quad \model, w_m \vDash_{a} \fOne \text{ or } \model, w_m \vDash_{a} \fTwo &\\
&\model, w_m \vDash_a \bigcirc \phi 
&&\quad \text{iff}  \quad \quad m \neq n \text{ and } \model, w_{m+1} \vDash_a \phi &\\
&\model, w_m \vDash_a \phi \ \Until \ \psi  
&&\quad \text{iff} \quad \quad \text{either } \model, w_m \vDash_a \psi, \text{ or for all $k$ s.t. } m \leq k < j  \;\model, w_k \vDash_a \phi \text{ and } \model, w_j \vDash_a \psi &\\
&\model, w_m \vDash_{a} \wbox_{q} \fOne 
&&\quad \text{iff} \quad \quad  
\frac{|\{w_l : l \leq m \mid \model, w_l \vDash_{a} \fOne\}|}{m} \ge q &\\
&\model, w_m \vDash_{a} \bbox_{q} \fOne 
&&\quad \text{iff} \quad \quad \text{there is a } a' \in \Pg,\;
\frac{|\{w_l : l \leq m \mid \model, w_l \vDash_{a'} \fOne\}|}{m} \ge q &\\
&\model, w_m \vDash_{a} \sat_{q} \fOne 
&&\quad \text{iff} \quad \quad 
\frac{|\{w_l : l \leq m \mid \model, w_l \vDash_{a} \fOne\}|}{n} \ge q &\\
&\model,w_{m} \vDash_{a} \wnext_{q}\fOne 
&&\quad \text{iff} \quad \quad m\neq n, a \in \Pg \text{ and }
\frac{|\{a ' \in \Pg: a' =_{p,i} a \mid \model, w_{m+1} \vDash_{a '} \fOne\}|}
{|\{a ' \in \Pg :  a' =_{p,i} a\}|} \ge q &\\
&\model, w_m \vDash_{a} \bstar_{q} \fOne 
&&\quad \text{iff} \quad \quad \text{there is a } w'\in W,\;
\frac{|\{a ' \in \Pg : \model, w' \vDash_{a '} \fOne\}|}{| \Pg|} \ge q. &\\
\end{align*}
where $=_{p,i}$ is just syntactic sugar abbreviating the following expression $\forall p, i(i \leq m \to (w_i \in a'(p) \leftrightarrow w_i \in a(p))$, intuitively asking whether two assignments are equivalent \emph{up to some world $i$}.
\end{definition}

\noindent
Notably $\vseg$ is nothing but an assignment (either in $\Pg$ or not), to be intuitively interpreted as the ``real'' or ``current'' observation.
%
%
%
%
%
The clause for atoms stresses that we are interpreting them as outcomes of the experiments composing the observation.
The clause for negation results from the assumption of mutual exclusivity of the atoms in $\At_\Omega$ and is preserved for composed formulas due to the axioms of the $\sigma$-algebras.
The clauses for conjunction and disjunction interpret these notions as the logical counterpart of intersection and union between events.
Modalities give a way to express the observed frequency of verified formulas or related notions. This ensures that the semantics is well defined. The set of operators, however, is not minimal: it is possible to define some operators in terms of others. For example, $\sat$ can be defined in terms of a combination of $\wbox$ and $\bigcirc$: since $\sat$ expresses the minimum attainable frequency of an event over the total observations in the series given the events so far, it counts both backwards and forward. The frequency of the event is expressed with a fraction $\frac{l}{n}$ (by straightforward application of Definition \ref{def:semantics}), where $l\leq m$. $l$ can be obtained with $\wbox$, which will give a fraction $\frac{l}{m}$, and the total number of events $n = |W|$ with subsequent conjunctions of $\bigcirc$ (the full formalization is left to the interested reader). Clearly, the size of such a formula explodes in the number of values our random variable can take, and the smaller $m$ is.

\bigskip
We can now re-consider informal examples presented in Section~\ref{sec:keyOp} in terms of the formal semantics defined above.

\paragraph{White box $\wbox$.}
The $w$- and $\vseg$-dependent operator $\wbox_{q}$ expresses that until world $w$, the argument formula $\fOne$ has happened \emph{at least} $q\in \mathbb{Q}_{[0,1]}$
of the time.

\begin{ex}\label{ex:wbox2}
    Let $\model=(\{w_1, \dots, w_n\}, \Rexp, \fr, \Pg,\vseg)$ and $m\in \{1,\dots, n\}$.
    In our logic, the following relation expresses that, given an observation $\vseg$, after flipping a coin $m$ times, the result has been head (expressed as $\head$) at least $\frac{l}{m}$ of the times,
    $$
    \model, w_m \vDash_{\vseg} \wbox_{\frac{l}{m}} \head.
    $$
\end{ex}

\noindent
The semantics of this operator can be rephrased as:
$\model, w_m \vDash_{a} \wbox_{q} \fOne$ iff $\frac{|\{ w \; : \; w_m\Rexp w \; \wedge \;  \model, w \; \vDash_{a} \fOne\}|}{|\{w_1, \dots, w_m\}|}  \ge q$.

\paragraph{Black box $\bbox$.} 
The $w$-dependent and $\vseg$-independent operator $\bbox_q$ expresses that $\fOne$ occurred up to world $w$ with a frequency of at least $q \in \mathbb{Q}_{[0,1]}$, measured against a fixed target frequency distribution $\Pg$ assigned to every atom.
Thus, this operator captures properties relating the target distribution $\Pg$ to the events in a series up to the current state.

\begin{ex}
    Let us assume that we think our coin is (empirically) unbiased, i.e.,~the relative frequency of $\head$ is equal to the relative frequency of $\tail$.
    Given a series of $n\ge2$ tosses, after the first toss 
    any outcome  
    is compatible with the desired distribution, i.e.,~after the first toss we can still have outcomes so that, in the end, the coin will come out to be (empirically) unbiased.
    For example, obtaining $\head$ ($\head$) in our first toss is compatible with our desired frequency, which is formalizable as
    $$
    \model, w_1 \vDash_{\vseg} \bbox_{1} \head
    $$
    where $\model$ is obtained extending with any $\vseg$ a structure $\str=(\{w_1, \dots, w_n\}, \Rexp, \fr, \Pg)$, such that  $\fr(\head)=\frac{1}{2}$, because there is an $a \in \Pg$ such that
    $$
    \model, w_1 \vDash_{a}\head
    $$
    For example, consider an assignment $a$ assigning  $\head$ to the first $\frac{n}{2}$ worlds and $\tail$ to the remaining worlds.  Clearly, according to $a$, $\head$ is satisfied in $w_1$ and is in $\Pg$, since $\mu_a(\head)=\frac{1}{2}$.
\end{ex}

\noindent
    In fact, this holds at any world with index at most $\frac{|W|}{2}$, since it is clear that for any sequence of results there is a situation compatible with the fact that the coin is (empirically) unbiased.\footnote{Since we assume that our coin is (empirically) unbiased, the set $W$ cannot consist of an odd number of worlds.}
    On the other hand, at the end of the series, i.e.~considering the last world $w_n$, the measure-index $q$ of the boxed formula for $\head$ must correspond at most to the target distribution.

\paragraph{Circle operator $\sat$.} The $w$- and $\vseg$-dependent operator $\sat _{q}$ expresses that, until $w$, the argument formula $\fOne$ has happened at least $q \times n$ times, where $n=|W|$ is the total number of trials of the series.
This can be interpreted as the minimum attainable frequency of the argument formula $\fOne$ at the conclusion of the series, based on the information about $\vseg$ available up to $w$.

\begin{ex}
    Going back to Example~\ref{ex:wbox2}, let us assume that, again, we have flipped the coin $m$ times, that $l$ times the outcome has been $\head$, and we know that we will flip the coin $n$ times. Then, we can formally express that, at the end of the experiment, we will obtain $\head$ with \emph{at least} probability $\frac{l}{n}$ as follows:
    $$
    \model, w_m \vDash_{\vseg} \sat_{\frac{l}{n}} \head.
    $$
    Intuitively, this operator provides information about the number of times a value has come out (in the concrete case considered, $\vseg$, and at a specific event $m$ in the series), taking into account the total number $n$ of trials constituting the series.
\end{ex}

While the white-box operator $\wbox$ can be interpreted as the $\vseg$-dependent counterpart of the black-box operator $\bbox$, conversely $\bbox$ can be defined in terms of the instances of $\wbox$ associated with all assignments $a$ compatible with a given target distribution $\Pg$.
Moreover, unlike $\wbox$, the operator $\sat$ takes into account the total number of trials constituting the series, rather than only those observed so far.
Indeed, the information provided by $\wbox$ is the frequency of the event \emph{until} a given point (i.e., time) $w$: it thus relates the number of times the phenomenon has been observed to the number of attempts experienced so far.
The operator $\sat$, on the other hand, relates the number of occurrences of the phenomenon under investigation to the total number of attempts in the whole series.

\paragraph{Next operator $\wnext$.} The $w$-dependent, but $\vseg$-dependent operator $\wnext_{q}$ expresses that, assuming that
$\vseg \in \Pg$, for $\Pg$ being the target (well-defined) distribution, the probability that $\fOne$ happens at the next world is at least $q$.
It is calculated as the ratio of assignments in $\Pg$ that satisfy $\fOne$ in the next world, assuming that assignments are equiprobable.

\begin{ex}\label{ex:next2}
    Let us consider an ``empirically unbiased'' coin and a series of four tosses $W=(w_1, w_2, w_3, w_4)$. If we know that the first two outcomes are $\tail$ (expressed by $\tail)$, i.e. $\{w_1,w_2\}\subseteq \vseg$, it will be certain that at the third toss the outcome is $\head$ (expressed as $\head$):
    $$
    \model, w_2 \vDash_\vseg \wnext_1 \head.
    $$
    Clearly, the fact that the series consists of four coin tosses is determined by $|W|$, the (empirical) fairness of the coin is expressed by $\Pg$ (both of which are specified by $\str$), and the outcome that the first two tosses end up as $\tail$ is specified by $\vseg$. The $\bigcirc$ operator can be seen as a stronger version of $\wnext_q$: it states that at the next point in the series its argument formula holds without referencing $\Pg$.
\end{ex}

\paragraph{Star operator $\bstar$.}
The $w$- and $\vseg$-independent star operator $\bstar_{q}$ expresses that, according to a given well-defined distribution $\Pg$, there is a world $w$ such that the argument formula $\fOne$ is satisfied there by at least $q$ of the assignments in $\Pg$.
%
For $\at$ atomic it holds that (see Theorem~\ref{theorem:AlternativeStar}): 
$$
\model, w \vDash_{\vseg} \star_{q} \at \quad \text{ iff } \quad \fr (\at) \ge q.
$$

\begin{ex}
    If the coin is ``empirically fair'', i.e.,~such that $\fr(\head)=\frac{1}{2}$, then, for any actual observation and at any time, we have probability (at least or) equal to $\frac{1}{2}$ 
    of observing $\head$:
    $$
    \model, w_m \vDash_{\vseg} \bstar_{\frac{1}{2}} \head.
    $$
    Similarly, if we are dealing with a different $\Pg''$ such that the coin is (empirically) biased and returns head only $\frac{1}{3}$ of the times, i.e.~$\fr'(\head)=\frac{1}{3}$, the following is true in any world and for any $\vseg$:
    $$
    \model, w_m \vDash_{a} \bstar_{q} \head
    $$
    being $q \leq \frac{1}{3}$.
\end{ex}

\bigskip
By convention, we define operators to denote that the measure of the argument formula is greater than or equal to a given value $q\in \mathbb{Q}_[0,1]$.
However, alternative relational constraints could also be used.
For readability, we treat these operators as abbreviations within the framework.

\begin{notation}
\label{notation:=}
    For any operator $\oplus \in\{\wbox,\bbox,\sat,\bstar,\wnext\}$, we let $\oplus _{<q}$  denote the strict inequality counterpart, where the measure of the set of interest is strictly less than $q\in \mathbb{Q}_{[0,1]}$. For example
    $$
    \model, w_m \vDash_{a} \wbox_{<q}\fOne \quad \text{iff} \quad \frac{|\{w : w_m \Rexp w \; \wedge \; \model, w \vDash_{a} \fOne\}|}{m} < q.
    $$
    Similarly, we use $\oplus_{=q}$ to define the operator expressing that the set of interest has measure precisely $q$. For instance:
    $$
    \model, w_m \vDash_{a} \wbox_{=q}\fOne \quad \text{iff} \quad \frac{|\{w : w_m \Rexp w \; \wedge \; \model, w \vDash_{a} \fOne\}|}{m} = q.
    $$
\end{notation}

\noindent
All of these operators are definable in terms of the primitive ones (the proof is omitted for brevity).

Observe that the operator $\oplus_{=q}$ behaves differently for white and black modalities.
In particular, with white modalities, the formula $\oplus_{=q}$ is satisfied by a unique probability value $q$.
The same does not hold for black modalities.
More formally:

\begin{theorem}
\label{Univocita=white}
    For every $\model = (W, \Rexp, \fr, \Pg, \vseg)$, assignment $a \in \Pg \cup \{\vseg \}$, 
    $\PL$-formula $\fOne$ and $q,q'\in \mathbb{Q}_{[0,1]}$, if $\oplus \in \{\wbox, \sat, \wnext\}$
    $$
    \model, w \vDash_{a} \oplus_{=q}
 \fOne \; \text{ and } \; \model, w \vDash_{a} \oplus_{=q'}
 \fOne \quad \Longrightarrow \quad q = q'.
 $$
\end{theorem}

\begin{proof}
    The proof follows trivially from the observation that for any $\vseg$, the truth condition for $\oplus_{=q} \fOne$ is defined in terms of a fraction that can have only one value.
    For example, the truth of a white-boxed formula in the world $w_m$ is determined by the \emph{unique} measure $\frac{|w_l : l \leq m \wedge \model, w_k \vDash_a \fOne\}|}{m}$ being equal to the given $q$.
\end{proof}

\noindent
This property fails for the black modality $\blacksquare$, as illustrated by Example~\ref{ex:NonUnicityBlack=}. The same failure applies to $\star$ when $\phi$ is a modal formula (see Example 5.8 below); however, when $\phi$ is propositional, $\star_{=q}\phi$ does satisfy uniqueness (see Lemma~\ref{lemma:mu} and Theorem~\ref{theorem:AlternativeStar}).

\begin{ex}
\label{ex:NonUnicityBlack=}
Consider a model such that $W=\{w_1,w_2,w_3\}$, $\At_\Omega = \{\head,\tail\}$, and $\Pg$ according to which $\fr(\head)= \frac{2}{3}$ and $\fr(\tail)= \frac{1}{3}$.
For any assignment $a \in \Pg \cup \{\vseg \}$, we have $\model, w_2 \vDash_{a} \bbox_{=1}\head$ and $\model, w_2 \vDash_{a} \bbox_{=\frac{1}{2}}\head$.
Indeed, both $a'$ and $a''$ such that~$a' (\head) = \{w_1,w_2\}$ and $a' (\tail) = \{w_3\}$, and $a'' $ such that~$a'' (\head) = \{w_1,w_3\}$ and $a'' (\tail) = \{w_2\}$ are in $\Pg$.
\end{ex}

\noindent
Indeed, the truth value of $\bbox_{=q} \fOne$ is defined in terms of the \textit{existence} of an assignment such that~the given fraction has value $q$, and the truth value of $\bstar_{=q} \fOne$ is defined in terms of the \textit{existence} of a world such that the given fraction has value $q$.
Hence, multiple $q$s can make the formula true.

In order to extend the property in Theorem~\ref{Univocita=white} to black modalities, we need to refer to the maximal applicable value.

\begin{notation}
\label{notation:top}
For any operator $\oplus \in \{\wbox,\bbox, \sat, \bstar, \wnext\}$, we use $\oplus_{\overline{q}}$ to define the operator that is true only for the greatest index that makes true its counterpart:
    $$
    \model, w_m \vDash_{a} \oplus_{\overline{q}}\fOne \quad \Longleftrightarrow  \quad 
    \model, w_m \vDash_{a} \oplus_{q}\fOne \: \land \: (q' >q  \: \Longrightarrow \: \model, w_m \nvDash_{a} \oplus_{q'}\fOne)
    $$
\end{notation}

\noindent
Again, these operators are definable in terms of primitive ones.
For white modalities, this notion coincides extensionally with that introduced in Notation~\ref{notation:=} to express the precise measure.

\begin{theorem}
\label{theorem:Top=White}
   For every $\model = (W, \Rexp, \fr, \Pg , \vseg)$, world $w\in W$, assignment $a \in \Pg \cup \{\vseg \}$, 
   $\PL$-formula $\fOne$ and $q\in \mathbb{Q}_{[0,1]}$, if $\oplus \in \{\wbox, \sat, \wnext\}$,
    $$
    \model, w \vDash_{a} \oplus_{=q}
 \fOne \quad \Longleftrightarrow \quad \model, w \vDash_{a} \oplus_{\overline{q}}
 \fOne .
 $$ 
\end{theorem}

\begin{proof}
    $(\Rightarrow)$
    From Definition~\ref{notation:=}, if $\model, w \vDash_{a} \oplus_{=q} \fOne$, then $\model, w \vDash_{a} \oplus_{q} \fOne$.
    Moreover, from Theorem~\ref{Univocita=white}, if $q' > q $, then $\model, w \nvDash_{a} \oplus_{q'} \fOne$.
    Hence, from Definition~\ref{notation:top}, $\model, w \vDash_{a} \oplus_{\overline{q}}\fOne$.
    $(\Leftarrow)$. The proof is by contraposition:
    if $\model, w \nvDash_{a} \oplus_{=q} \fOne$ and $\model, w \nvDash_{a} \oplus_{q} \fOne$, then $\model, w \nvDash_{a} \oplus_{\overline{q}}\fOne$.
    If $\model, w \nvDash_{a} \oplus_{=q} \fOne$ and $\model, w \vDash_{a} \oplus_{q} \fOne$, then there is $q'> q $ such that $\model, w \vDash_{a} \oplus_{q'}\fOne$ ($q'$ being the value of the fraction used in the truth condition for $\oplus_{q} \fOne$).
    But, then, $\model, w \nvDash_{a} \oplus_{\overline{q}}\fOne$.
\end{proof}

\noindent
Moreover, in this case Theorem~\ref{Univocita=white} can be generalized to black modalities.

\begin{theorem}
\label{UnivocitaMaxblack}
    For every $\model = (W, \Rexp, \Pg ,\fr, \vseg)$, world $w\in W$, assignment $a \in \Pg \cup \{\vseg \}$, 
    $\PL$-formula $\fOne$ and pair $q,q'\in \mathbb{Q}_{[0,1]}$, if $\oplus \in \{\wbox,\bbox, \sat, \bstar, \wnext\}$,
    $$
    \model, w \vDash_{a} \oplus_{\overline{q}}
 \fOne \; \land \; \model, w \vDash_{a} \oplus_{\overline{q'}}
 \fOne \quad \Longrightarrow \quad q = q'.
 $$
\end{theorem}

\begin{proof}
    The proof 
    follows from Definition~\ref{notation:top} and the complete ordering of rational numbers.
\end{proof}

\begin{sloppypar}
Via the outcome-atom correspondence $\iota:\mathfrak{P}_\Omega\longrightarrow \Omega$, we extend $\iota$ compositionally to propositional formulas over $\mathfrak{P}_\Omega$, treating atoms in $\mathfrak{P}\setminus\mathfrak{P}_\Omega$ as $\bot$, consistently with $a(p)=\varnothing$ for such $p$:
 $S_p:=\{\iota(p)\}$,
 $S_{\neg\phi}:=\Omega\setminus S_\phi$,
 $S_{\phi\land\psi}:=S_\phi\cap S_\psi$, and
   $S_{\phi\lor\psi}:=S_\phi\cup S_\psi$.
Since $\Omega$ is finite and discrete, every $S_\phi$ is measurable, and we set $\mu_\Phi(\phi):=\mu_\Phi(S_\phi)$.
\end{sloppypar}

\begin{lemma}\label{lemma:mu}
    For any 
    model over a well-defined $\Pg$ associated with $\Omega$, any propositional formula $\fOne$ and given fixed $w\in W:$
    $$
        \frac{|\{a \in \Pg \; : \; \model, w \vDash_{a} \fOne\}|}{|\Pg|} = \fr(\fOne).
    $$
\end{lemma}
\begin{proof}[Proof sketch]
For an atom $p$, the equality is established directly. 
For compound $\phi$, the result follows automatically from the compositionality of $\models_a$ (Definition~\ref{def:semantics}), which gives $\{a\in F_\Omega:\mathbb{M},w\models_a\phi\land\psi\}=\{a\in F_\Omega:\mathbb{M},w\models_a\phi\}\cap\{a\in F_\Omega:\mathbb{M},w\models_a\psi\}$, together with the additivity of $\mu_\Phi$ as a measure over $2^\Omega$ (see Definition~\ref{def:prob}) applied to $S_{\phi\land\psi}=S_\phi\cap S_\psi$; similarly for $\lor$ we have $S_{\phi\lor \psi}=S_{\phi}\cup S_\psi$ and for $\neg$, it holds that $S_{\neg\phi}=2^\Omega \setminus S_{\phi}$.
\end{proof}

\noindent
The following theorem straightforwardly follows from Lemma~\ref{lemma:mu}.
\begin{theorem}
\label{theorem:AlternativeStar}
Given a model $\model = (W,\Rexp,\fr, \Pg,\vseg)$, for any $a\in \Pg \cup \{\vseg\}$, any propositional formula $\fOne$ and any world $w\in W$:
\begin{align*}
    \model , w \vDash _{a} \bstar_{q} \fOne \quad &\text{ iff } \quad \fr(\fOne) \ge q
    \\
    \model , w \vDash _{a} \bstar_{=q} \fOne \quad &\text{ iff } \quad \fr(\fOne) = q
\end{align*}
\end{theorem}

\noindent
Intuitively, Theorem~\ref{theorem:AlternativeStar} proves that, given a propositional formula $\fOne$, the following procedures give the same result:
\begin{itemize}
\itemsep0em
    \item counting the assignments in $\Pg$ that satisfy $\fOne$ in a given world $w_m$ (notice that the choice of the world is irrelevant, since we assume that the events in $\Omega$ are mutually independent);
    \item counting the worlds in $W$ satisfying $\fOne$, given an assignment $a \in \Pg$ (notice that the choice of the assignment is irrelevant, since all assignments in $\Pg$ are compatible with frequency $\fr (\fOne)$).
\end{itemize}
%

The following example shows that neither Lemma~\ref{lemma:mu} nor Theorem~\ref{theorem:AlternativeStar} generalizes to modal argument formulas.

\begin{ex}
\label{ex:StarFreq}
    Consider a structure such that~$W=\{w_1,w_2,w_3\}$, $\At_\Omega = \{\head,\tail\}$, and $\Pg$ with $\fr(\head)= \frac{2}{3}$ and $\fr(\tail)= \frac{1}{3}$.
    For any $\vseg$ allowing to define $\model$ and every assignment $a ' \in \Pg$, we have $\model, w_3 \vDash_{a'} \bbox_{=\frac{1}{3}}\tail$ and so, for every $a \in \Pg \cup \{\vseg\}$, it holds that $\model, w_3 \vDash_{a} \bstar _{1} \bbox_{=\frac{1}{3}}\tail$.
    However, for every assignment $a ''\in \Pg$, $\model, w_1 \nvDash_{a ''} \bbox_{=\frac{1}{3}}\tail$ and $\model, w_2 \nvDash_{a''} \bbox_{=\frac{1}{3}}\tail$.
    Hence, for every assignment $a ''\in \Pg$, the relative frequency of worlds that satisfy the formula $\bbox_{=\frac{1}{3}}\tail$ in $W$ is $\frac{1}{3}$, and so for every $a \in \Pg \cup \{\vseg\}$, we have $\model, w_3 \nvDash_{a} \bbox _{1} \bbox_{=\frac{1}{3}}\tail$, while $\model, w_3 \vDash_{a} \bstar _{1} \bbox_{=\frac{1}{3}}\tail$.
\end{ex}

\noindent
On the other hand, for propositional formulas, Theorems~\ref{Univocita=white} and \ref{theorem:Top=White} can be generalized to $\bstar_{=q}$:

\begin{theorem}
\label{theorem:=TopStarProp}
    For every $\model = (W, \Rexp, \fr, \Pg , \vseg)$, assignment $a \in \Pg \cup \{\vseg \}$, propositional formula $\fOne$ and $q,q'\in \mathbb{Q}_{[0,1]}$, then: 
    \begin{align*}
    \model, w \vDash_{a} \bstar_{=q}
 \fOne \; \text{ and } \; \model, w \vDash_{a} \bstar_{=q'}
 \fOne \quad &\; \Longrightarrow \quad q = q' \\
 \model, w \vDash_{a} \bstar_{=q}
 \fOne \quad &\Longleftrightarrow \quad \model, w \vDash_{a} \bstar_{\overline{q}}
 \fOne
\end{align*}
\end{theorem}

\begin{proof}
    Both results follow trivially from Theorem~\ref{theorem:AlternativeStar}. 
\end{proof}

\bigskip
Finally, relying on Definition~\ref{def:semantics}, we can easily characterize validity in a model, validity in a structure and validity \textit{simpliciter}. 

\begin{definition}\label{df:validity}
A $\PL$-formula $\fOne$ is said to be:
\begin{itemize}
\itemsep0em
\item \emph{valid in a model $\model$} iff $\model = (W, \Rexp, \Pg, \vseg)$ and for every $w\in W$ it holds that $\model, w \vDash_{\vseg} \fOne$;
\item \emph{valid in a structure $\str$} iff $\str = (W, \Rexp, \Pg) $ and for every assignment $\vseg$, $\fOne$ is valid in the model $\model =(W, \Rexp, \Pg, \vseg)$ obtained adding $\vseg$ to $\str$;
\item \emph{valid (\text{simpliciter})} iff for every structure $\str = (W, \Rexp, \fr, \Pg)$, $\fOne$ is valid in it.
\end{itemize}
We use $\model \vDash_{\vseg} \fOne$, $\str \vDash \fOne$ and $\vDash \fOne$ to denote, respectively, validity in a model, in a structure and \textit{simpliciter}.
\end{definition}

\subsection{Interaction between modal and temporal operators}

A further aspect of the expressive power of $\PL$ is due to the combination of frequency modalities and temporal operators from standard $\LTLf$. These interactions allow us to further express distributional properties of the series throughout its temporal evolution. Let us, for now, consider the interaction of $\wbox$ with $\bigcirc$ and $\Until$. 
\begin{ex}
     Let's consider again the case in Example \ref{ex:wbox} with results HHTT. The nesting $\bigcirc(\wbox_q\phi)$ intuitively expresses that, at the next point in the series, $\phi$ will have occurred at least $q$ of the times. Then, clearly, for a model $\model = (\{w_1, w_2, w_3, w_4\}, \Rexp, \fr, \Pg, \vseg)$, we have for example that 
     $$
     \model, w_1 \vDash_a \bigcirc (\wbox_1\head)     $$
     Since in $w_2$ it simply holds that $\wbox_1\head$. The other possible nesting $\wbox_q(\bigcirc\phi)$, instead, expresses that the formula $\bigcirc\phi$ has occurred at least $q$ of the times in the series. Then, considering the same model $\model$ as above, we have that
     $$
     \model, w_2 \vDash_a \wbox_\frac{1}{2} (\bigcirc\head)
     $$
     Since the only world where $\bigcirc\head$ holds is $w_1$.  
\end{ex}

\begin{ex}
Consider again the scenario in Example \ref{ex:wbox} (HHTT). We have that $\head \Until \tail$ holds at every world $w_m$ of the series, directly following from Definition \ref{def:semantics}. Also consider the formula $\wbox_q(\head \Until \tail)$. Intuitively, it states that, at the time of evaluation, $\head \Until \tail$ has held at least $q$ of the times. Since, in the scenario of \ref{ex:wbox}, $\head \Until \tail$ holds at every world, then we have, for example, that
$$
\model, w_2 \vDash_a \wbox_1(\head \Until \tail)
$$
since both $w_1$ and $w_2$ satisfy the $\Until$ condition, being $\tail$ realized in $w_3$. 
Now, let's consider another possible nesting, namely $(\wbox_q\head) \Until \tail$. Intuitively and semantically, it has a different meaning from the previous example: it states that $\head$ holds $q$ of the times until $\tail$ holds (in this case, without specifying the frequency of $\tail$, though of course a combination of two $\wbox_q$ operators is possible). In the scenario we are considering, it holds that
$$
\model, w_2 \vDash_a (\wbox_1 \head) \Until \tail
$$
since, at $w_2, \head$ has occurred with frequency 1 and later $\tail$ holds. 
\end{ex}
Notice that the combination of the $\bigcirc$ and $\Until$ operator with any of the modal operators implies that the evaluation looks both \emph{backwards} and \emph{forward} in the series. As the frequency operators ($\wbox_q, \circledcirc_q, \wnext_q$) count the frequency of an event up until the evaluation point in the sequence, they are all \emph{backward-looking}. The temporal operators $\bigcirc$ and $\Until$, instead, need to be evaluated on the following states of the sequence, making them \emph{forward-looking}. The combination of the two allows us to express whether, for example, the distributional properties that have held so far will still hold in the future states of the sequence, e.g. in  the combination $(\wbox_q\phi)\Until\psi$ we have seen in the example above. 
Let us now briefly consider the interaction with other operators as well. 

\begin{ex}
    Consider a scenario in which we assume our coin is fair, i.e., that $\fr(\head)=\frac{1}{2}$. Then, at any moment, the probability of obtaining $\head$ at the next toss is clearly $\frac{1}{2}$. This simple fact can be expressed with the formula $\star_\frac{1}{2} \bigcirc \head$, which reads exactly as its intuitive meaning (notice, however, that it doesn't hold at the last world $n$, where the formula is not defined). 
\end{ex}
\noindent
The opposite nesting, while being equally expressive, denotes a different semantic flavour: $\bigcirc\star_q\phi$ states that, at the next time point, $\phi$ will have a global probability of at least $q$ of holding. This property is time-independent, but it is expressed as holding at the next point in time, and is equivalent to stating that it holds in the present, and thus to $\star_q\phi$.  Let us conclude with a simple property of the interaction between $\bigcirc$ and other operators: 
\begin{remark}\label{remark:WboxSatQAddition}
If $\model, w_m\vDash_a \wbox_{=\frac{l}{m}} \phi$ and $\model, w_m\vDash_a \bigcirc \phi$ then $\model, w_{m+1} \vDash_a \wbox_{=\frac{l+1}{m+1}} \phi$. This is easy to see, as every new occurrence of an event increases its relative frequency over the current number of trials by $\frac{{l+1}}{m+1}$ (with $l$ being the number of worlds that satisfy $\phi$), by straightforward application of Definition \ref{def:semantics}. Similarly, if $\model, w_m\vDash_a \sat_{=\frac{l}{n}} \phi$ and $\model, w_m\vDash_a \bigcirc \phi$ then $\model, w_{m+1} \vDash_a \sat_{=\frac{l+1}{n}} \phi$, as the relative frequency of $\phi$ simply increases by 1. 
\end{remark}
\noindent
Finally, notice that the interactions between temporal and modal operators described above are not exhaustive, and while many different properties could be investigated, the focus of this work is more toward those of the frequency operators.

\section{Expressing properties of series in $\PL$}\label{sec:properties}

In this section, we focus on the expressive power of $\PL$, showing that this formalism allows us to express naturally all the properties from our guiding example (Section~\ref{sec:guide}).
For readability, all subsequent sections begin by recalling the problem we are going to formalize in $\PL$.

\subsection{Checking the frequency of an outcome in a series}


\probwhite*

\noindent
First, we consider a special case, the one of Example~\ref{exwhite}, with results $HHT$.
This can be expressed by the simple formula
$\wbox_{\frac{2}{3}}\head$, which is clearly satisfied in the third world $w_3$ of the model $\model = (\{w_1, w_2, w_3, w_4\}, \Rexp, \fr, \Pg, \vseg)$.
More generally, we can express the (partial) relative frequency of any event $\fOne$ (atomic or compound) in any world, or compare it with any given threshold.

\begin{notation}
\label{not:partialfreq}
Given $\model = (W, \Rexp, \fr, \Pg, \vseg)$, $w_m \in W$, and a propositional formula $\fOne$, we let:
$$
\mu_{\tu a}(l, \fOne) := \frac{|\{w_l : l \le m \wedge \model, w_l \vDash_\vseg \fOne\}|}{m}
$$
denote the ($\vseg$-dependent) \emph{partial relative frequency} of $\fOne$ up to $w_m$. 
Of course, it respects all the axioms defining a probability measure. 
Notice that $\mu_{\tu a}$ is, in general, distinct from $\fr$, the latter is the target frequency fixed by the structure via $\Pg$, independent of the observation $\vseg$, whereas $\mu_{\tu a}$ depends on $\vseg$ and on the prefix length $m$.
\end{notation}

\begin{theorem}
\label{prop:frequencyWBox}
    For any $\model = (W, \Rexp,\fr, \Pg,\vseg)$, with $| W| =n$, $w_m\in W$, and $q\in \mathbb{Q}_{[0,1]}$:
    \begin{itemize}
        \itemsep0em
        \item[i.] $\model, w_m \vDash_\vseg \wbox_q \fOne$ iff $\mu_{\tu a}(m,\fOne)\ge q$;
        \item[ii.] $\model, w_m \vDash_\vseg \wbox_{=q} \fOne$ iff $\mu_{\tu a}(m,\fOne)= q$.
    \end{itemize}
\end{theorem}
\begin{proof}[Proof Sketch]
    Straightforward application of Definition~\ref{def:semantics} and Notation~\ref{not:partialfreq}. Observe that this property is entirely $\Pg$-independent.
\end{proof}

\noindent

\subsection{Checking partial frequencies compatible with the target one}

We consider Problem~\ref{probblack}, which requires checking which partial frequencies of events are compatible with the target frequency $\Pg$ at a given time.

\probblack*

\noindent
An instance of this problem is presented in Example~\ref{exblack}, which illustrates that given four coin tosses such that~the frequency of \(\head\) is \(\frac{1}{2}\), it is always true that at the third toss the observed frequency of \(\head\) could be either \(\frac{1}{3}\) or \(\frac{2}{3}\).
This can be expressed by the formulas \(\bbox_{\frac{1}{3}} \head\) and \(\bbox_{\frac{2}{3}} \head\), which are true in \(w_3\) for any model defined by four worlds and with associated frequency $\fr(\head)=\frac{1}{2}$.

Notice that, while $\bbox$ is explicitly designed to express that the frequency of a formula is compatible with $\Pg$, it does not automatically express compatibility of frequencies of \textit{more than one formula} with $\Pg$.
Indeed, given a model $\model$, a formula $\bbox _q\fOne$ is satisfied in a world $w_m\in W$ iff there is $a \in \Pg$ such that it satisfies $\fOne$ in at least $q \times m$ worlds preceding $w_m$.
However, the formulas $\bbox _q\fOne $ and $ \bbox _{q'}\fOne'$ can both be satisfied in a world $w_m\in W$ (or equivalently $\bbox _q\fOne \land \bbox _{q'}\fOne'$ can be satisfied there) even if no $a \in \Pg$ satisfies $\fOne$ in at least $q \times m$ worlds preceding $w_m$ and $\fOne'$ in at least $q'\times q'$ worlds preceding $w_m$.
Consider the following example:

\begin{ex}
\label{ex:BBboxCOmpatibi}
\begin{sloppypar}
    Consider a model $\model$ of coin tosses 
    so defined that $W=\{w_1 , w_2,w_3,w_4\}$ and such that $\fr (\head)=\fr (\tail) = \frac{1}{2}$.
    Regardless of $\vseg$, $\model, w_2 \vDash_{\vseg} \bbox_{1}\head$ (consider the assignment $a_1 \in \Pg $ such that~$a_1(\head)= \{w_1,w_2\}$ and $a_1(\tail)= \{w_3,w_4\}$). 
    Moreover, $\model, w_2 \vDash_{\vseg} \bbox_{1}\tail$; consider the assignment $a_2 \in \Pg $ such that~$a_2(\tail)= \{w_1,w_2\}$ and $a_2(\head)= \{w_3,w_4\}$).
    Hence, $\model, w_2 \vDash_{\vseg} \bbox_{1}\head\land \bbox_{1}\tail$.
    However, there is clearly no assignment $a \in \Pg$ such that~all worlds preceding $w_2$ satisfy both $\head$ and $\tail$.
\end{sloppypar}
\end{ex}

However the use the black-box alone is not sufficient to express the existence of an assignment compatible with $\Pg$ and with the world-dependent frequency determined by argument formulas.
While the verification of the sentence $\wbox _{q}\fOne \land \wbox _{q'}\fOne'$ in a world $w_m$ and according to an assignment $\vseg$ expresses that, according to $\vseg$, $q \times m$ worlds accessible from $w_m$ verify $\fOne$ and $q' \times m$ worlds accessible from $w_m$ verify $\fOne'$, the verification of the sentence $\bbox _{q}\fOne \land \bbox _{q'}\fOne'$ in $w_m$ 
\textit{does not} say that there is an $a \in \Pg$ such that~$q \times m$ worlds accessible from $w_m$ verify $\fOne$ and $q' \times m$ worlds accessible from $w_m$ verify $\fOne'$ according to $a$.
The verification of this formula only ensures the weaker condition that there is an $a \in \Pg$ such that $q \times m$ worlds accessible from $w_m$ verify $\fOne$ and there is a, possibly different, $a ' \in \Pg$ such that~$q' \times m$ worlds accessible from $w_m$ verify $\fOne'$ according to $a'$.
In order to express compatibility of partial frequencies of \textit{more than one formula} with $\Pg$, we need to take into account nested modalities of different kinds.

\begin{ex}
\label{ex:ConjunctionOfBbox}
\begin{sloppypar}
    Consider the model described in Example~\ref{ex:BBboxCOmpatibi}.
     Even though $\model , w_1 \vDash_{\vseg} \bbox_{\frac{1}{2}}\head \land \bbox_{\frac{1}{2}}\tail$, clearly there can be no assignment (and so no assignment in $\Pg$) that gives $\head$ at least half of the times in the worlds accessible from $w_1$ \emph{and} $\tail$ at least half of the times in the worlds accessible from $w_1$.
    This fact can be expressed as follows, by relying on nested modalities:
    $$
    \model , w_1 \nvDash_{\vseg} \bbox_{\frac{1}{2}}(\wbox_{\frac{1}{2}}\head \land \wbox_{\frac{1}{2}}\tail)
    $$
    \noindent
    Here, the order of the modalities is relevant, since $\model , w_1 \vDash_{\vseg} \wbox _{\frac{1}{2}} (\bbox_{\frac{1}{2}}\head \land \bbox_{\frac{1}{2}}\tail)$.
    On the contrary, there is an assignment in $\Pg$ satisfying $\head$ in the worlds accessible from $w_2$ at least half of the times and $\tail$ the other half.
    Moreover, the same holds when considering the world $w_4$. 
    More formally,
    $$
    \model , w_2 \vDash_{\vseg} \bbox_{\frac{1}{2}}(\wbox_{\frac{1}{2}}\head \land \wbox_{\frac{1}{2}}\tail) \quad \text{ and } \quad
    \model ,
    w_4 \vDash_{\vseg} \bbox_{\frac{1}{2}}(\wbox_{\frac{1}{2}}\head \land \wbox_{\frac{1}{2}}\tail)
    $$
        For any bigger $W$, the same holds in all worlds $w_{2l}$ for $l\in \Nat$.
    In addition, in the worlds accessible from $w_4$ (the last world) all the assignments in $\Pg$ satisfy both $\head$ and $\tail$ at least half of the times.
    Being $q> \frac{1}{2}$, this property is formally expressed as follows:
    $$
    \model , w_4 \vDash_{\vseg} \neg \bbox_{q}\head \land \neg \bbox_{q}\tail.
    $$
\end{sloppypar}
\end{ex}

These observations can be generalized (see Lemma~\ref{conj:lcd(mu(p))} (see also Lemma~\ref{lemma:mu2}), showing how to express the existence of assignments in $\Pg$ satisfying a distribution of frequency for atoms in the accessible worlds, and how to find worlds in which such a distribution can be realized.

\subsection{Checking compatibility of an observation with a desired frequency}\label{subsec:CheckCompatibility}

In this section, we formalize the compatibility of an observation with a target frequency.

\probOne*

\noindent

Specifically, we can use either $\bbox$ and $\wbox$, or $\bstar$ and $\sat$, or $\wbox$ and $\bstar$ to check that at a time $w$, a given assignment $\vseg$ is \emph{not} consistent with a given $\Pg$.

\begin{theorem}\label{prop:incons}
For any model $\model$, associated with $W$ and $\vseg$, world $w_m\in W$, and any propositional formula $\fOne$, the following implications hold:
\begin{align}
\model ,w_m \vDash _{\vseg} \wbox_{q} \fOne \land \neg \bbox_{q} \fOne 
\quad &\Longrightarrow \quad
\vseg \not\in \Pg \\
\model ,w_m \vDash _{\vseg} \sat_{q} \fOne \land \neg \bstar _{q} \fOne  \quad &\Longrightarrow \quad \vseg \not\in \Pg \\
%
%
\model ,w_m \vDash_{\vseg} \wbox_{l/m} \fOne \land \bigcirc \fOne, \quad mq+1>n\mu_{\Phi}(\fOne) \quad &\Longrightarrow \quad \vseg \not\in \Pg \\
\model ,w_m \vDash_{\vseg} \sat_{\frac{l-1}{n}} \fOne \land \bigcirc \fOne \land \neg \star_{\frac{l}{n}}\fOne \quad &\Longrightarrow \quad \vseg \not\in \Pg 
\end{align}
where in (3), we assume that $\fr(\fOne) = \frac{l'}{n}$.
\end{theorem}
\begin{sloppypar}
\begin{proof}[Proof Sketch]
\begin{enumerate}
    \item By straightforward application of Definitions~\ref{def:semantics}:
        $\model, w_m \vDash_\vseg \wbox_q \fOne \wedge \neg \bbox_q \fOne$
        iff 
        \Big($\model, w_m \vDash_\vseg \wbox_q \fOne \text{ and }
        \model, w_m \vDash_\vseg \neg \bbox_q \fOne$\Big)
        iff
        \Big($\frac{|\{w_k : l \leq m \wedge \model, w_l \vDash_\vseg \fOne\}|}{m} \ge q$ and 
        $\model, w_m \not\vDash_\vseg \bbox_q \fOne$\Big)
        iff $\frac{|\{w_k : l \leq m \wedge \model, w_l \vDash_\vseg \fOne\}|}{m} \ge q$         and there is no %
        $a \in \Pg$ such that~$        \frac{|\{w_k : l \leq m \wedge \model, w_l \vDash_a \fOne\}|}{m} \ge q$. We conclude $\vseg \not \in \Pg$.
        \item By contraposition, assume $\tu a \in \Pg$. 
    Then, for any $q\in \mathbb{Q}$, either (a) $\fr(\fOne)\ge q$, so that $\model, w_m \vDash_{\tu a} \star_q \phi$
    or (b) $\fr(\fOne) < q$, so that, since $\tu a \in \Pg$, for any $w_m$, $\frac{|\{w_l : l \leq m \wedge \model, w_l \vDash_{\tu a} \fOne\}|}{n} \leq \frac{|\{w_l : l \leq n \wedge \model, w_l \vDash_{\tu a} \fOne\}|}{n} < q$.
    Thus, in both cases, by Definition~\ref{def:semantics}, $\model, w_m \not\vDash_{\tu a} \sat_q \fOne \wedge \neg \star_q \fOne$.
    \item Assume, by contradiction, that $\vseg \in \Pg$ so that $\fr(\fOne)= \frac{r}{n}$ for some $r\in \Nat$.
    By  Definition~\ref{def:semantics}, the first condition implies that $\fOne$ holds in at least $m q+1=l+1$ worlds;
    that is, according to $\vseg$, the frequency of $\fOne$ is \emph{at least} $\frac{l+1}{n}$.
    By the second condition, we know that $m q + 1 = l+1 > n \fr(\fOne)= r$.
    Thus, $\frac{l+1}{n}> \frac{r}{n} = \mu(\fOne)$, and $\vseg \not \in \Pg$ against the assumption.
\item  By straightforward application of Remark \ref{remark:WboxSatQAddition}, Definition~\ref{def:semantics} and (2): the combination of $\sat$ and $\bigcirc$ implies that $\model, w_{m+1} \vDash_\vseg \sat_{\frac{l}{n}}\fOne \land \neg \star_{\frac{l}{n}}\fOne$. Thus, the case is identical to (2) and the proof is the same.
\end{enumerate}
\end{proof}
\end{sloppypar}
\noindent
Intuitively, this allows us to check consistency between $\Pg$ and $\vseg$ \emph{inside the model}.

The converse direction, i.e.,~that $\vseg$ \emph{is} compatible with given distribution $\Pg$, holds in a weaker sense only: while we cannot ensure that $\vseg \in \Pg$ until all worlds have been checked, we can show that an assignment $\vseg$ is compatible with $\Pg$ \textit{up to a world $w_m$}.

\begin{definition}\label{def:compatiUpTo}
    For every model $\model = (W, \Rexp, \fr, \Pg, \vseg)$ and world $w_m\in W$, we say that $\mathbf{a}$ is compatible with $F_\Omega$ up to $w_m$ iff there is $a\in F_\Omega$ such that $\forall p\in\mathfrak{P}_\Omega\,\forall i\le m\,(w_i\in\mathbf{a}(p)\leftrightarrow w_i\in a(p))$
\end{definition}

\noindent
Due to this notion, we can strengthen Theorem~\ref{prop:incons} in the other direction. 
Preliminarily, notice that by the finite nature of structures and the semantic definitions of operators $\sat$ and $\bstar$, we can always define the probability of the argument formula in terms of a fraction $\frac{l}{n}$, for $l,n\in \Nat$, where $n$ is the cardinality of the corresponding model.

\begin{prop}\label{prop:auxCons}
Let $\model=(W, \Rexp, \fr, \Pg,\vseg)$ be such that $|W|=n$ and associated with a given sample space $\Omega$. 
For any atom associated with it $\at\in \At_\Omega$ whose correspondence probability is different from 0, it holds that:
    \begin{itemize}
        \itemsep0em
        \item there is an $l\in \{1,\dots, n\}$ such that~for any $w\in W$ 
        $\model, w \vDash_{\vseg} \bstar_{=\frac{l}{n}} \at$;
        \item for any $w_i\in W$, there is an $l'\in \{1,\dots, n\}$ such that~$\model, w_i \vDash_{\vseg} \sat_{=\frac{l'}{n}}\at.$
    \end{itemize}
\end{prop}

\begin{theorem}\label{prop:cons}
    For any model $\model=(W, \Rexp, \fr, \Pg, \vseg)$ such that $|W|=n$ and associated with a given $\Omega$ then, up to some $w_m\in W$, $\vseg$ is compatible with $\Pg$ iff for every atom $\at\in \At_\Omega$ there are $l,r\in \Nat$, with $l\leq r$ and $r-l\leq n-m$, such that
    $$
    \model, w_m \vDash_{\vseg} \sat_{=\frac{l}{n}}\at
    \wedge
    \bstar_{=\frac{r}{n}}\at.
    $$
\end{theorem}
\begin{proof}[Proof Sketch]
    Straightforward consequence of Definition~\ref{def:semantics} and Proposition~\ref{prop:auxCons}.
\end{proof}
\noindent

\begin{sloppypar}
Thus, we can formalize our problem by considering a model $\model = (W, \Rexp, \fr, \Pg, \vseg)$, where $|W|=n$ and $\Pg$ corresponds to the mentioned frequency distribution. Intuitively, we check whether, up to $m$, $\vseg$  is compatible with $\Pg$ by checking the satisfiability of formulas defined in Theorem~\ref{prop:cons}, i.e.~whether for any $\at\in \At_\Omega$, there are $l,r\in \{0,\dots, n\}$ such that
$$
\model, w_m \vDash_{\vseg} \sat_{=\frac{l}{n}}\at
    \wedge
    \bstar_{=\frac{r}{n}}\at 
$$
where $l\leq r$ and $r-l \leq n-m$.

\begin{ex}
Going back to Example~\ref{ex1}, we can show that the first observation, i.e.,~$\tail\tail$ is compatible with the desired frequency $\Pg$. 
In this case, $\At_\Omega=\{\head, \tail\}$, $
\model =(\{w_1, w_2, w_3, w_4\}, \Rexp, \Pg,\vseg)$, where $\Pg$ is such that~$\fr(\head)=\frac{1}{2}$, $\fr(\tail)=\frac{1}{2}$, 
and $\vseg$ is such that~$w_1,w_2, w_3\in \vseg(\tail)$.
Therefore, by Definition~\ref{def:semantics},
$
\model, w_2 \vDash_{\vseg} \circ_{=\frac{2}{4}} \tail.
$
Moreover, for any $w\in W$ and $a\in \Pg\cup \{\vseg\}$, $\model, w \vDash_{a} \bstar_{=\frac{1}{2}} \tail$, so that \textit{a fortiori}
$
\model, w_2 \vDash_{\vseg} \circ_{=\frac{1}{2}} \tail \wedge \bstar_{=\frac{1}{2}} \tail.
$
Since the same holds for $\head$, by Theorem~\ref{prop:cons}, we conclude that $\vseg$ is compatible with $\Pg$, up to $w_2$.
On the other hand, when we move to $w_3$, we have
$
\model, w_3 \vDash_{\vseg} \circ_{\frac{3}{4}}\tail
$
but clearly, for any $w\in W$,
$
\model, w \vDash_{\vseg} \neg \bstar_{\frac{3}{4}} \tail.
$
Therefore, 
$
\model, w_3 \vDash_{\vseg} \circ_{\frac{3}{4}}\tail \wedge \neg \bstar_{\frac{3}{4}} \tail.
$
So, by Theorem~\ref{prop:incons}, $\vseg$ is not compatible with $\Pg$.
\end{ex}
\end{sloppypar}
\noindent
Clearly, Theorem~\ref{prop:cons} is stronger than Theorem~\ref{prop:incons}.
In particular, the following holds:

\begin{theorem}
    For any $\model =(W, \Rexp, \fr, \Pg, \vseg )$, $w_m\in W$, $l,r\in \Nat$ such that $l\leq r$ and $r-l\leq n-m$, and $q\in \mathbb{Q}_{[0,1]}$
    $$
    \model, w_m \vDash_{\vseg} \sat_{=\frac{l}{n}}\at \wedge \bstar_{=\frac{r}{n}} \at \quad \Longrightarrow \quad \model, w_m \not\vDash_{\vseg} \sat_{q} \at \wedge \neg \bstar_{q} \at. 
    $$
\end{theorem}

\begin{proof}
    Suppose, towards a contradiction, that for some $q'\in \mathbb{Q}_{[0,1]}$, $\model, w_m \vDash_{\vseg} \sat_{q'} \at \wedge \neg \bstar_{q'} \at$.
    If there is such $q'$, then the formula is satisfied when $q'$ is substituted with $\frac{l}{n} $ as well, due Theorem~\ref{theorem:Top=White}
    and the trivial observation that if $\neg \bstar_{q'} \at$ is satisfied in a world $w$ and $q''\geq q'$, then also $\neg \bstar_{q''} \at$ is satisfied in $w$.
        By Theorems~\ref{theorem:AlternativeStar} and~\ref{theorem:=TopStarProp}, the following, respectively, hold: $\fr (\at)=\frac{r}{n}$ and $\model, w_m \vDash_{\vseg} \bstar_{q} \at$ iff $q \leq \frac{r}{n}$.
    Hence, $\frac{r}{n} < \frac{l}{n}$ and so $r < l$. 
\end{proof}

\subsection{Expressing the probability of realizing
a desired frequency}

In this section, we continue the previous problem by addressing the realization of a frequency compatible with the expected one.

\probTwo*

\noindent

\begin{notation}
Given $\vseg$ and $a \in \Pg \cup \{\vseg\}$, we call $P^{\vseg}_{w_m}$ the probability we are looking for is:
$P(a \in \Pg | \forall i \leq m (a (w_i)=\vseg (w_i)) ).$
\end{notation}

\noindent
When all outcomes are equiprobable, this probability can be defined in terms of formulas expressing the compatibility of the assignment $\vseg$ with $\Pg$ (see Section~\ref{subsec:CheckCompatibility}).

\begin{theorem}
\label{conj:probabilityFrequence}
    For any model $\model=(W, \Rexp, \fr, \Pg,\vseg )$ associated with a given $\Omega$ and such that~$|W|=n$, atom $p\in \At_\Omega$, and world $w_m \in W$, if the following formula holds
    $
    \model, w_m \vDash_{\vseg} \bigwedge_i^{\at_i\in \At_\Omega} \sat_{=\frac{l_i}{n}}\at_i
    \wedge
    \bstar_{=\frac{r_i}{n}}\at_i 
    $
    and the outcomes of $\Omega$ are equiprobable, then the probability $P^{\vseg}_{w_m}$ that, at world $w_m$,  $\vseg$ is compatible with $\Pg$ is:
    $
    P^{\vseg}_{w_m}=\frac{\prod ^{| \At_\Omega|}_{i=1}{{n - m - \sum ^{i-1}_{j=1} (r_j - l_j)}\choose{r_i - l_i}}}{| \At_\Omega|^{n-m}}.
    $
\end{theorem}

\begin{proof}
    If the values of $\Omega$ are equiprobable, then all assignments are equiprobable.
    If all assignments are equiprobable, then the probability that $\vseg$ is in $\Pg$ is the ratio of the number of assignments in $\Pg$ behaving like $\vseg$ in the first $m$ worlds over the total number of all assignments behaving like $\vseg$ in these worlds.
    To count the number of assignments in $\Pg$, which behave like $\vseg$ in the first $m$ worlds, we consider how many ways there are to assign worlds after $w_m$ to the atoms so that the frequencies in $\Pg$ are realized.
The expression
$
\prod ^{| \At_\Omega|}_{i=1}{{n - m - \sum ^{i-1}_{j=1} (r_j - l_j)}\choose{r_i - l_i}}
$
computes this.
In particular, each factor $n - m - \sum ^{i-1}_{j=1} (r_j - l_j)$ represents the worlds after $w_m$ left open by the assignment of the atoms $\at_1,\dots,p_{i-1}$, and $r_i - l_i$ is the number of worlds to be still assigned to $\at _i$.
On the other hand, to count the total number of assignments behaving like $\vseg$ in the wolds preceding $w_m$, we consider how many ways there are \textit{in general} to assign worlds after $w_m$ to the atoms in $\At_\Omega$.
This is precisely what the expression
$
| \At_\Omega|^{n-m}
$
computes.
\end{proof}

\begin{ex}\label{ex:probabilitaPg}
For $|\At_{\Omega}|=2$, the 
expression for the probability that 
$\vseg$ is 
compatible with $\Pg$ is  simplified as follows:
$
\frac{{{n-m}\choose{r_1 - l_1}}\times {{n-m - (r_1-l_1)}\choose{r_2 - l_2}}}{2^{n-m}}.
$
Let $\model =(W, \Rexp, \fr, \Pg,\vseg)$ be such that~$|W|= 2$, $\At_\Omega= \{\head,\tail\}$, $\fr(\head)=\fr(\tail)=\frac{1}{2}$ and $\vseg$ be so defined that $w_1 \in \vseg(\head)$. At $w_1$ the probability that $\vseg$ is compatible with $\Pg$ is $\frac{1}{2}$.
This is expressed in our language as follows:
$$
\model, w_1 \vDash_{\vseg} (\sat_{=\frac{1}{2}} \head \wedge \bstar_{=\frac{1}{2}}\head) \wedge
(\sat_{=0} \tail \wedge \bstar_{=\frac{1}{2}}\tail).
$$
Indeed, as desired,
$
P^{\vseg}_{w_1} = \frac{{{2-1}\choose{0}}\times {{2-1 - 0}\choose{1}}}{2} = \frac{1}{2}.
$
Consider also $\model' = (W, \Rexp, \fr', \Pg ',\vseg)$, with $\Pg'$ such that~$\fr(\head)=1$ and $\fr(\tail)=0$. 
The probability that $\vseg$ is compatible with $\Pg'$ is expressed by the formula:
$$
\model, w_1 \vDash_{\vseg} (\sat_{=\frac{1}{2}} \head \wedge \bstar_{=1}\head) \wedge
(\sat_{=0} \tail \wedge \bstar_{=0} \tail)
$$
and, as desired,
$
P^{\vseg}_{w_1} =\frac{{{1}\choose{2 - 1}}\times {{1- 1}\choose{0}}}{2} = \frac{1}{2}.
$
As a more complex example, consider the model 
$\model''
= (W, \Rexp, \Pg,\vseg)$ with $|W| = 6 $, and, again, $\fr (\head)=\fr (\tail)= \frac{1}{2}$.
As for $\vseg$, assume $w_1\in\vseg (\head)$ and $w_2,w_3\in\vseg (\tail)$
(the rest of the assignment is irrelevant).
The probability that $\vseg$ ends in $\Pg$ evaluated at $w_3$ is $\frac{3}{8}$, since there are only three assignments in $\Pg$ 
assigning $w_1$ at $\head$, and both $w_2$ and $w_3$ at $\tail$, and only eight assignments in general behaving like this on the first three worlds.
This is in compatible with our Theorem, since
$
    \model'', w_3 \vDash_{\vseg} (\sat_{=\frac{1}{6}}\head
    \wedge
    \bstar_{=\frac{3}{6}}\head)
    \wedge
    (\sat_{=\frac{2}{6}}\tail
    \wedge
    \bstar_{=\frac{3}{6}}\tail)
    $
and
$
    P^{\vseg}_{w_3} =\frac{{{6 - 3}\choose{3 - 1}}{{6 - 3 -  (3 - 1)}\choose{3 - 2}}}{2^{6-3}}=
    \frac{{{3}\choose{2}}{{1}\choose{1}}}{2^{3}}=
    \frac{{3}}{8}.
$
\end{ex}

\noindent
Theorem~\ref{conj:probabilityFrequence} can be generalized to cases in which the outcomes in $\Omega$ are \textit{not} equiprobable.
However, we still need the assumption 
that all outcomes in $\Omega$ are independent, that is that the satisfaction of $\at \in \At_\Omega$ at world $w_m$ is not affected by what happens in the rest of the series.

\begin{theorem}
\label{conj:probabilityFrequenceGeneral}
    For any model $\model=(W, \Rexp, \Pg,\vseg )$ associated with a given $\Omega$ and such that~$|W|=n$, world $w_m \in W$, if the following holds
    $$
    \model, w_m \vDash_{\vseg} \bigwedge_i^{\at_{\omega_i}\in \At_\Omega} \sat_{=\frac{l_i}{n}}\at_{\omega_i}
    \wedge
    \bstar_{=\frac{r_i}{n}}\at_{\omega_i} 
    $$
    and the values of $\Omega$ are all independent,
    then the probability $P^{\vseg}_{w_m}$ that $\vseg$ is compatible with $\Pg$ evaluated at world $w_m$ is:
    $
    P^{\vseg}_{w_m}=\prod ^{| \At_\Omega|}_{i=1}{{n - m - \sum ^{i-1}_{j=1} (r_j - l_j)}\choose{r_i - l_i}}\times \prod ^{| \At_\Omega|}_{i=1}{P(\at _i)^{r_i-l_i}}.
    $
\noindent
Where $P(\at_{\omega_i})$ is the probability associated with outcome $\omega_i$.
\end{theorem}

\begin{proof}[Proof sketch]
    The proof is similar to the one for Theorem~\ref{conj:probabilityFrequence}, see Appendix~\ref{sec:appendix}.
\end{proof}

\begin{ex}\label{ex:probabilitaPgGeneral}
Let us consider again Example~\ref{ex:probabilitaPg}.
This time, let us assume that $\head$ and $\tail$ are not equiprobable, but that $P(\head)=\frac{2}{3}$ and $P(\tail)=\frac{1}{3}$.
As for the first model, it still holds that
$
\model, w_1 \vDash_{\vseg} (\sat_{=\frac{1}{2}} \head \wedge \bstar_{=\frac{1}{2}}\head) \wedge
(\sat_{=0} \tail \wedge \bstar_{=\frac{1}{2}}\tail).
$
However, this time the probability of ending with the distribution of frequency in $\Pg$ is:
$
P^{\vseg}_{w_1} = {{2-1}\choose{0}}\times {{2-1 - 0}\choose{1}}\times P(\head)^0 \times P(\tail)^1 = \frac{1}{3}.
$
Indeed, although $w_2$ can still satisfy either $\tail$ or $\head$, the probability that $\tail$ is satisfied (the only one leading to an assignment in $\Pg$) is $\frac{1}{3}$.
We skip the second model (which is a trivial variation of the first) and focus on the third one, keeping the previous assumptions on the probabilities of $\head$ and $\tail$.
It still holds that
$
    \model'', w_3 \vDash_{\vseg} (\sat_{=\frac{1}{6}}\head
    \wedge
    \bstar_{=\frac{3}{6}}\head)
    \wedge
    (\sat_{=\frac{2}{6}}\tail
    \wedge
    \bstar_{=\frac{3}{6}}\tail).
    $
However, this time the probability of ending with the frequency distribution compatible with $\Pg$ is:
$
P^{\vseg}_{w_3} ={{6 - 3}\choose{3 - 1}}\times{{6 - 3 -  (3 - 1)}\choose{3 - 2}}\times P(\head)^2 \times P(\tail)^1 =
    \frac{{4}}{9}.
$
Notice that, not surprisingly, this probability is higher than that for a fair coin, since we want more $\head$ than $\tail$, and this result shows that our coin is biased.
\end{ex}

\subsection{Expressing the probability of subsequent outcomes}

Finally, we consider the problem of estimating at one point of the series the probability of an outcome in the (relative) future, given that we know its global frequency.

\probThree*

\noindent
Given a model $\model = (W, \Rexp, \fr, \Pg,\vseg)$, the probability that we are looking for is the following one:
$$
\Pr(\model , w_{m+1}\vDash_a \fOne \; | \; a \in \Pg \land  \forall i \leq m (a (w_i)=\vseg (w_i)) ).
$$
Our operator $\wnext$ is explicitly designed to compute this.

\begin{theorem}
\label{theorem:probabilityNext}
    For any model $\mathbb{M}=(W,\Rexp,\fr, F_\Omega,\mathbf{a})$, world $w_m\in W$, $q\in\mathbb{Q}_{[0,1]}$, $m\neq n$ and $\vseg \in F_\Omega$ and formula $\fOne$,
$$
\model ,w_m\models_{\mathbf{a}}\triangleright_q\phi \iff \Pr\big(\model,w_{m+1}\models_{a'}\phi \ \mid\ a'\in F_\Omega \ \land\ \forall i\le m\,(a'(w_i)=\mathbf{a}(w_i))\big) \ge q,
$$
where $a'$ ranges over assignments, and the conditional probability is the uniform probability over $\{a'\in F_\Omega : \forall i\le m\,(a'(w_i)=\mathbf{a}(w_i))\}$. Moreover, $\model,w_m\models_{\mathbf{a}}\triangleright_{=q}\phi$ iff that probability is exactly $q$.
\end{theorem}

\begin{proof}
    As for Theorem~\ref{conj:probabilityFrequence}, the probability is obtained by counting the ratio of assignments that satisfy given restrictions.
    In this case, we count the assignments in $\Pg$ which behave like $\vseg$ in the first $m$ worlds and satisfy $\fOne$ in the world $w_{m+1}$ on top of the assignments in $\Pg$ behaving like $\vseg$ in the first $m$ worlds.
    Notice that, since we count only assignments in $\Pg$, and these assign the same number of worlds to the atoms in $\At_\Omega$, in this case, we do not need to assume that the outcomes in $\Omega$ are equiprobable.
\end{proof}

\noindent
For atomic formulas, we can state the previous result in the following, more elegant, manner:

\begin{theorem}
\label{theorem:probabilityNextAtom}
    For any model $\model = (W, \Rexp, \fr, \Pg,\vseg)$, world $w_m\in W$ with $m\neq n$, $\vseg \in \Pg$, $q\in \mathbb{Q}_{[0,1]}$, and $\at_i \in \At_\Omega$,
    $$
    \model, w_m \vDash_\vseg \bigwedge_i^{\at_i\in \At_\Omega}  \wnext_{=q_i} \at_i
    $$
    iff the probability that $ \model, w_{m+1} \vDash_\vseg \at_i$ proviso that $\vseg\in \Pg$ is exactly $q_i$.
\end{theorem}

\noindent
To illustrate this point more concretely, we recall and extend Example~\ref{ex3}, while formalizing it in our formal language.
\begin{ex}
Let us assume that Tiresias, who foresees the future, told us that today Ulysses will toss the coin four times and obtain $\tail$ half of the time.
%
Moreover, let us assume that Ulysses has thrown the coin the first time, obtaining $\tail$.
We can represent this situation using $\model = (W, \Rexp, \Pg,\vseg)$ with $|W|=4$, $\Omega = \{\head,\tail\}$, $\Pg$ such that~$\fr(\head)=\fr(\tail)=\frac{1}{2}$, and $w_1\in\vseg (\tail)$.
If we trust Tiresias' prediction, we will express the probability of the next toss as follows:
$$
\model, w_1 \vDash_\vseg \wnext_{q} \head \wedge \wnext_{q'
} \tail
$$
\begin{sloppypar}
To calculate the values of $q$ and $q'$, we use the semantics of Definition~\ref{def:semantics},
    $\model, w_1 \vDash_{\vseg} \wnext_{q
    } \head$ iff 
    $
    \frac{|\{ a \; : \; a \in \Pg \; \wedge \; w_1 \in a(\tail) \; \wedge \; w_2 \vDash_{a} \head\}|}{|\{a \; : \; a \in \Pg \wedge w_1 \in a(\tail)\}|} = \frac{2}{3} \ge q
    $ %
Indeed, the only assignments $a',a'',a''' \in \Pg$ for $|W|=4$ and such that~$w_1$ is assigned to $\tail$ are precisely the ones such that either: (i) $w_2,w_3\in a'(\head)$ and $w_4\in a'(\tail)$ or (ii) $w_2, w_4 \in a''(\head)$ and $w_3\in a '' (\tail)$ or (iii) $w_3,w_4 \in a'''(\head)$ and $w_2 \in a'''(\tail)$.
These are the three and only possible assignments  compatible with $\Pg$. 
Additionally, if we also impose that $w_2$ is assigned to $\head$, 
clearly only the former two assignments are acceptable.
A similar story could be told for $\model, w_1 \vDash_{\vseg} \wnext_{q'
    } \tail$, leading to $q'\leq \frac{1}{3}$.
Then, Ulysses tosses the coin for the second time and, again, we have $\tail$. Due to Tiresias' premonition, we know that then, we will always obtain $\head$. 
Indeed, the following relation holds in our model:
$$
\model, w_2 \vDash_\vseg \wnext_1 \head \wedge \wnext_0 \tail
$$
Notably this situation essentially corresponds to Example~\ref{ex3}, which can be formalized as done here with Tiresias's prediction. 
\end{sloppypar}

We can consider more complicated formulas. For example, after the first toss, we already knew that \emph{if} the second toss had returned $\tail$, then the third toss would have been $\head$ with (at least) probability 1, while \emph{if} the second toss had returned $\head$, then the third toss would have been $\head$ with (at least) probability $\frac{1}{2}$.
Formally, this is obtained by considering the same model $\model$ and by showing that the following satisfiability relation holds:
$$
\model, w_1 \vDash_\vseg \wnext_1((\tail \wedge \wnext_{1} \head) \vee (\head \wedge \wnext_{\frac{1}{2}} \tail))
$$
\begin{sloppypar}
\noindent
By Definition~\ref{def:semantics}, this formula holds when
$
    \frac{|\{a : a \in \Pg \wedge w_1 \in a(\tail) \wedge \model, w_2 \vDash_{a} (\tail \wedge \wnext_{1} \head)  \lor  (\head \wedge \wnext_{\frac{1}{2}} \tail)\}|}{|\{a : a \in \Pg \wedge w_1 \in a(\tail)\}|} \ge 1.
$
The set at the numerator contains the assignments $a\in \Pg$ such that~$\model, w_1 \vDash_{a} \tail$ and either $\model, w_2 \vDash_{a} \tail$ and $\model, w_2 \vDash_{a}  \wnext_1 \head$, or $\model, w_2 \vDash_{a} \head$ and $\model, w_2 \vDash_{a}  \wnext_{\frac{1}{2}} \tail$.
Let us see that this set contains three assignments:
\end{sloppypar}
\begin{itemize}
\begin{sloppypar}
    \item The assignments such that~$\model, w_1 \vDash_{a} \tail$, $\model, w_2 \vDash_{a} \tail$ and $\model, w_2 \vDash_{a}  \wnext_1 \head$ are those $a \in \Pg$ such that $w_1,w_2\in a (\tail)$ and
$$
\frac{|\{a' : a' \in \Pg \wedge \forall \at, i(i\leq m \to (w_i \in a '(\at) \leftrightarrow w_i \in a (\at)) \wedge w_3\in a'(\head)\}|}{|\{a' : a' \in \Pg \wedge \forall \at, i(i\leq m \to (w_i \in a '(\at) \leftrightarrow w_i \in a (\at)) \}|} \ge 1 
$$
It should be obvious that there is only one assignment in $\Pg$ such that $w_1,w_2\in a (\tail)$, and it satisfies also the third condition, since it assigns $w_3$ (and $w_4$) to $\head$;
\end{sloppypar}
\begin{sloppypar}
\item On the other hand, the assignments such that $\model, w_1 \vDash_{a} \tail$, $\model, w_2 \vDash_{a} \head$ and $\model, w_2 \vDash_{a}  \wnext_{\frac{1}{2}} \tail$ are those $a \in \Pg$ such that $w_1\in a (\tail)$, $w_2\in a (\head)$, and
$$
\frac{|\{a' : a' \in \Pg \wedge \forall \at, i(i\leq m \to (w_i \in a '(\at) \leftrightarrow w_i \in a (\at)) \wedge w_3\in a'(\tail)\}|}{|\{a' : a' \in \Pg \wedge \forall \at, i(i\leq m \to (w_i \in a '(\at) \leftrightarrow w_i \in a (\at)) \}|} \ge \frac{1}{2} 
$$
There are two assignments in $\Pg$ such that~$w_1\in a (\tail)$, $w_2\in a (\head)$: the first one assigns $w_3$ to $\head$ and $w_4$ to $\tail$, while the second one does the opposite.
Hence, both of them satisfy condition, since only one of two assigns $w_3$ to $\tail$. 
\end{sloppypar}
\end{itemize}
As for the set at the denominator, it contains all the assignments $a\in \Pg$ such that $\model, w_1 \vDash_{a} \tail$.
Hence, since there are only three such assignments (which assign respectively $w_2$, $w_3$, or $w_4$ to $\tail$), the fraction gives exactly $1$ as result.

\end{ex}

A natural question regards how to calculate the values of $q$ that make $\wnext _{q}\fOne$ satisfiable in a world $w_m$, given an assignment $\vseg$.
The following result focuses on atomic formulas:
\begin{theorem}
\label{theorem:ValueNext=}
Given a model $\model = (W, R_E, \fr, \Pg , \vseg)$, such that $\vseg \in \Pg$, let $W^{\vseg}(\fOne)$ denote the set $\{w\in W : w \vDash _{\vseg} \fOne \}$ of worlds that satisfy $\fOne$ according to assignment $\vseg$.
Moreover, let us use ${\Rexp} _{w_m}$, with $w\neq n$, for the set $\{w\in W : w_m \Rexp w \}$ of worlds accessible from $w_m$ according to $\Rexp$, and ${\Rexp} _{w_m}^{\vseg}(\fOne)$ for the set $\{w\in W : w_m \Rexp w \land w\vDash \fOne \}$ of worlds accessible from $w_m$ according to $\Rexp$ and that satisfy $\fOne$.
For every model $\model$, world $w_m\in W$, and $\at_x\in \At_\Omega$ (where $x$  is an element of $\Omega$), it holds that:
$$
\model, w_m \vDash _{\vseg} \wnext _{=q} \at_x \quad \text{ iff } \quad  q = \frac{\fr (p_x)\times | W | - | {\Rexp} _{w_m}^{\vseg}(\at_x) |}{| W | - | {\Rexp} _{w_m} | }
$$
\end{theorem}

\section{Conclusion and future work}\label{sec:conclusion}

In this paper, we introduced $\PL$, a temporal logic designed to reason about properties of finite series of experiments. While this system is in continuity with the inspiration and the approach of various other temporal and probabilistic logics, our approach is original in its ability to explicitly formalize the verification of observed frequencies against a target theoretical distributions within a single temporal structure. By evaluating frequencies of events while they evolve, this logic provides a straightforward tool for expressing and studying the properties of complex systems as they change over time.

The verification tasks our logic satisfies concern: the frequencies of outcomes; the compatibility of partial frequencies with a target global frequency; the lowest attainable frequency of the argument formula at the conclusion of a series; the probability of specific future outcomes based on realized data and theoretical distributions; and the probability that a formula is satisfied at a given moment, assuming a baseline distribution. Beyond the properties formalized by the individual modalities, their interplay allows us to characterize both the possibility and the probability of ultimately attaining a target global frequency from an experimental partial one.

Several formal problems remain open, among these: the study of model checking for $\PL$; developing and investigating a proof system, as well as proof search algorithms and their underlying decision problems; defining a meaningful notion of validity in a context where temporal operators interact with resource, counting, or probabilistic constraints, and a corresponding notion of local validity; the interaction with well-known probabilistic and (e.g.,~\cite{Bacchus,FHM})  counting logics (e.g.,~\cite{ADLP}) or recent probabilistic temporal formulations like \cite{MaMP25}, as well as the associated complexity of deciding formulas within these frameworks, which may prove complete for classes like $\mathtt{PSAT}$~\cite{Papadimitriou}, $\mathtt{\#SAT}$~\cite{Valiant}, or $\mathtt{MAJSAT}$~\cite{Wagner}. 

In terms of application, our temporal framework offers promising avenues for dynamic bias mitigation in AI. Unlike static fairness metrics or one-time checks during training, temporal logic can allow the design to formally specify and monitor time-dependent fairness requirements throughout the lifecycle of a system to prevent compounding feedback loops. This leads to the need for developing a branching-time version of this semantics to dynamically intercept and correct biased decisions at inference time before they propagate, see \cite{Buda}.

\bibliographystyle{abbrv}
\bibliography{bibliography}

\newpage
\appendix
\section{Proofs from Section~\ref{sec:properties}}\label{sec:appendix}

The following results refer only to atoms, but are easily generalizable to non-atomic formulas.

\begin{lemma}
\label{conj:lcd(mu(p))}
Let $\model$ be a model and $a  \in \Pg \cup \{\vseg\}$ be an assignment.
For every world $w_m \in W$ (with $m\in \Nat)$, there is an $a '\in \Pg$ that assigns to each atom $\at$ exactly $\fr (\at) \times m$ worlds accessible from $w_m$ iff it holds that
$$
\model , w_m \vDash _{a} \bbox _{q}\bigg(\bigwedge_{\at\in \At}\wbox_{\fr(\at)}\at\bigg)
$$
with $q = (\text{LCD}_{\at \in\At_\Omega}(\fr(\at)))^{-1}$.
Notice that this happens every time $m$ is a common multiple of all $\fr(\at)^{-1}$ for $\at \in\At_\Omega$.
\end{lemma}


\begin{proof}
\begin{sloppypar}
        The formula $\bigwedge_{\at\in\At_\Omega}\wbox_{\fr(\at)}p$ (henceforth the $\land\wbox$-formula) is verified in a world iff the frequency of every atom in the previous worlds is the one prescribed by $\fr$.
    Indeed, if all atoms are verified with at least frequency $\fr$, since
$
\sum _{\at \in\At_\Omega} \fr(\at) = 1
$, no atom $\at$ can be verified with a frequency strictly greater than $\fr (\at)$,  
    Hence, $\bbox _{q}\big(\bigwedge_{\at\in\At_\Omega}\wbox_{\fr(\at)}p\big)$ (henceforth the $\bbox$-formula) counts the number of worlds before the actual one such that~the worlds accessible from them satisfy the distribution of frequency prescribed by $\fr$, 
    according to at least one $a ' \in \Pg$
    .
\\
	Then, the theorem is proved by induction on $m$ (the label of the world at which we evaluate the formula).
	The intuitive idea is as follows.
	If $m < \textrm{LCD}_{\at \in\At_\Omega}(\fr(\at))$, then the $\bbox$-formula is verified only for $q= 0$.
	Indeed, with less than $\textrm{LCD}_{\at \in\At_\Omega}(\fr(\at))$ worlds it is not possible to verify the $\land\wbox$-formula. 
	The first world in which $q$ can be greater than $0$ is $w_m$ with $m=\textrm{LCD}_{\at \in\At_\Omega}(\fr(\at))$, and in this case the greatest $q$ for which the $\bbox$-formula is satisfied is 
$
\big(\textrm{LCD}_{\at \in\At_\Omega}(\fr(\at))\big)^{-1} = \frac{1}{\textrm{LCD}_{\at \in\At_\Omega}(\fr(\at))}.
$
Indeed, there is only one world that verifies the $\land\wbox$-formula: the last one. 
	When we move to the next world, the greatest $q$ decreases, since there is one more world accessible and still only one that satisfies the $\land\wbox$-formula. 
	The greatest value of $q$ for which the formula is provable decreases until we reach a multiple of $\textrm{LCD}_{\at \in\At_\Omega}(\fr(\at))$ (that is another common denominator of all $\fr (\at)$ for $\at \in\At_\Omega$).
	In this case, the greatest value of $q$ for which the $\bbox$-formula is satisfied is still $(\textrm{LCD}_{\at \in\At_\Omega}(\fr(\at)))^{-1}$.
	This holds for all the multiples of $\textrm{LCD}_{\at \in\At_\Omega}(\fr(\at))$, since there are $i \times \textrm{LCD}_{\at \in\At_\Omega}(\fr(\at))$ worlds accessible, of which exactly $i$ satisfy the $\land\wbox$-formula.
	\\
	Let us now consider the inductive step:
	assume that the statement holds for $w_m$ and prove it for $w_{m+1}$.
	We have two cases:
	\begin{itemize}
	\item If the $\bbox$-formula holds in $w_m$ for $q = \big(\textrm{LCD}_{\at \in\At_\Omega}(\fr(\at))\big)^{-1}$, then $m$ is a common denominator of all $\fr(\at)$ for $\at \in\At_\Omega$, and the $\land \wbox$-formula is verified in $w_m$ as well.
	Moreover, there is no $a\in \Pg$ for which the $\land \wbox$-formula is verified in $w_{m+1}$ (assuming that, for no $\at \in\At_\Omega$, $\fr (\at) = 1 $).
	Hence, the greatest $q$ for which the $\bbox$-formula holds in $w_{m+1}$ is strictly smaller than $(\textrm{LCD}_{\at \in\At_\Omega}(\fr(\at)))^{-1}$, and the statement of the theorem holds vacuously.
	\item If the $\bbox$-formula does not hold in $w_m$ for $q = (\textrm{LCD}_{\at \in\At_\Omega}(\fr(\at)))^{-1}$, then $m$ is not a common denominator of all $\fr(\at)$ for $\at \in\At_\Omega$, and the $\land \wbox$-formula is not verified at $w_m$.
	By IH, in every previous world $w_i$, being $i$ a common denominator of all $\fr(\at)$ for $\at \in\At_\Omega$, the $\bbox$-formula holds in $w_m$ for $q = (\textrm{LCD}_{\at \in\At_\Omega}(\fr(\at)))^{-1}$.
	In particular, this holds for $j= \textrm{LCD}_{\at \in\At_\Omega}(\fr(\at))$.
	Now, let us consider two sub-cases:
	\begin{itemize}
		\item If $m+1$ is not a common denominator of all $\fr(\at)$ for $\at \in\At_\Omega$, then the statement of the theorem is still vacuously true for $w_{m+1}$;
		\item If $m+1$ is not a common denominator of all $\fr(\at)$ for $\at \in\At_\Omega$, then it is a multiple of $j = \textrm{LCD}_{\at \in\At_\Omega}(\fr(\at))$.
		Let it be $ m+1 = j \times l$.
		Since all the common denominators of all $\fr(\at)$ are multiples of $j$, before $w_{m+1}$ there are $j \times l$ worlds that verify the $\land \wbox$-formula for at least one $a '\in \Pg$.
		Hence, the greatest $q$ for which the $\bbox$-formula is verified is		
		$
		\frac{l}{(\textrm{LCD}_{\at \in\At_\Omega}(\fr(\at)))\times l} = (\textrm{LCD}_{\at \in\At_\Omega}(\fr(\at)))^{-1}.
		$
	\end{itemize}	
	\end{itemize}
\end{sloppypar}
\end{proof}

\begin{proof}[Proof of Theorem~\ref{conj:probabilityFrequenceGeneral}]
    First, we count the possible ways to assign worlds after $w_m$ to the atoms so that the frequencies in $\Pg$ are realized. This part of the proof is proved exactly as before.
    Then, we obtain the probability of each of these ways of completing the series.
    This is easier than it might seem because the atoms verified in the worlds $w_{m+1},\ldots, w_n$ are still the same.
    We just need to verify $r_1-l_1$ times the atom $\at_{\omega_1}$, and we know that the probability that this happens is $P(\at_{\omega_1})^{r_1-l_1}$.
    The same applies to the other atoms.
    Hence, the probability of each of these endings of the series is
$
    \prod ^{| \At_\Omega|}_{i=1}{P(\at _i)^{r_i-l_i}}.
$
    The desired result is obtained by summing together the probabilities of each of these endings, that is by multiplying this probability for the number of such endings.
\end{proof}

In Theorem~\ref{conj:lcd(mu(p))}, the existence of an assignment that assigns to each atom $\at$ exactly $\fr (\at) \times m$ worlds accessible from $w_m$ is expressed by the satisfaction at $w_m$ of the formula $\bbox _{q}(\bigwedge_{\at\in\At_\Omega}\wbox_{\fr(\at)}\at)$ for a suitable value of $q\in \mathbb{Q}_{[0,1]}$.
Alternatively, one may avoid fixing a specific value of $q$, and demonstrate the existence of such an assignment by observing that there exists some $q$ for which this formula holds at $w_m$, while it fails at both the preceding and the following worlds.
In other words, the value of $q$ for which the formula is satisfied exhibits a peak at $w_m$ if and only if there exists an assignment that assigns to each atom $\at$ exactly $\fr(\at)\times m$ worlds accessible from $w_m$.

\begin{lemma}\label{lemma:mu2}
Given a model $\model$ associated with $W$ and an assignment $a  \in \Pg \cup \{\vseg\}$, for every world $w_m \in W$ (with $m\in \Nat$), there is an assignment $a '\in \Pg$ that assigns each atom $\at\in \At_\Omega$ exactly to $\fr (\at) \times m$ worlds accessible from $w_m$ iff there is a $q$ such that
$$
\model , w_m \vDash _{a} \bbox _{q}\bigg(\bigwedge_{\at\in\At_\Omega}\wbox_{\fr(\at)}\at\bigg)
$$
\noindent
while
$$
\model , w_{m-1} \nvDash _{a} \bbox _{q}\bigg(\bigwedge_{\at\in\At_\Omega}\wbox_{\fr(\at)}\at\bigg) \quad \text{ and } \quad 
\model , w_{m+1} \nvDash _{a} \bbox _{q}\bigg(\bigwedge_{\at\in\At_\Omega}\wbox_{\fr(\at)}\at\bigg)
$$
for every $w_m$ for which $m-1$ and $m+1$ are well-defined, consistently with our structure definition.
\end{lemma}

\begin{proof}
\begin{sloppypar}
    The proof follows from Theorem~\ref{conj:lcd(mu(p))}, by observing that:
    \begin{itemize}
    \itemsep0em
\item $(\textrm{LCD}_{\at \in\At_\Omega}(\fr(\at)))^{-1}$ is the greatest possible value of $q$ for which the $\bbox$-formula is satisfiable; this gives the \textit{if} direction;
\item for $q\neq (\textrm{LCD}_{\at \in\At_\Omega}(\fr(\at)))^{-1}$, if 
$
\model , w_{m} \vDash _{a} \bbox _{q}\big(\bigwedge_{\at\in\At_\Omega}\wbox_{\fr(\at)}\at\big)
$
then also $\model , w_{m-1} \vDash _{a} \bbox _{q}\big(\bigwedge_{\at\in\At_\Omega}\wbox_{\fr(\at)}\at\big)$; this gives \textit{only if} direction.
\end{itemize}
\end{sloppypar}
\end{proof}

Finally, the following Theorem generalizes the previous results to non-atomic formulas, allowing to express the existence of an assignment in $\Pg$ which satisfies a distribution of frequency for more than one formula in the accessible worlds.

\begin{lemma}\label{theorem:frequencyMoreFormulas}
Given a model $\model$ and an assignment $a  \in \Pg \cup \{\vseg\}$, for every world $w_m \in W$, there is an $a '\in \Pg$ that satisfies the formulas $\fOne_1,\ldots, \fOne_i$ in at least 
$q_1 \times m, \ldots , q_i \times m$ (with $q_1,\ldots , q_i \in \mathbb{Q} $) worlds accessible from $w_m$ iff for some $q' \in \mathbb{Q} $ it holds that:
$$
\model , w_m \vDash _{a} \sat_{=q'}(\fOne \lor \neg \fOne)\land \bbox_{\frac{1}{m}} \Big(\sat_{=q'}(\fOne \lor \neg \fOne)\land \bigwedge^i_{j=1}\wbox_{q_j}\fOne_j \Big)
$$    
\end{lemma}
\begin{proof}
    Since $\fOne \lor \neg \fOne$ is satisfied at every world of the series, and $\sat$ expresses the ratio of the previous worlds satisfying its argument with respect to all the worlds in the series, we can use $\sat_{=q'}(\fOne \lor \neg \fOne)$ to fix the world looked at by $\bbox$ and $\wbox$.
    Indeed, given a structure, two assignments $\vseg'$ and $\vseg''$, and two worlds $w_j,w_k \in W$, there is $q' \in \mathbb{Q}_{[0,1]}$ such that $\model, w_j \vDash_{\vseg'} \sat_{=q'}(\fOne \lor \neg \fOne)$ and $\model , w_k\vDash_{\vseg''} \sat_{=q'}(\fOne \lor \neg \fOne)$ iff $j=k$.
    Hence, by requiring that there is a $q'$ such that~$\sat_{=q'}(\fOne \lor \neg \fOne)$ is satisfied both in $w_m$ and in the world satisfying $\wbox_{=q} \fOne$, we essentially impose that $l=m$.
    In this way, we know that $\bbox$ possibly changes the assignment $a$ with one in $\Pg$, but it does not change world.
    Hence, $\bigwedge^i_{j=1}\wbox_{q_j}\fOne_j$ is true in $w_m$ according to an assignment in $\Pg$.
    But this means exactly that the desired distribution of frequency is realized, because of Theorem~\ref{prop:frequencyWBox}.
\end{proof}

\noindent
By using $\wbox_{=q}$ we can strengthen the previous result, by requiring the existence of an assignment in $ \Pg$ that satisfies the formulas $\fOne_1,\ldots, \fOne_i$ in, respectively, (exactly) $q_1 \times m, \ldots , q_i \times m$, with $q_1,\ldots , q_i \in \mathbb{Q} $, worlds accessible from $w_m$.

\begin{proof}[Proof of Theorem~\ref{theorem:ValueNext=}]
\begin{sloppypar}
We just have to prove that 
$
\frac{|\{a  :{a }\in \Pg \; \wedge \; \forall \at, i(i\leq m \to (w_i \in a (\at) \leftrightarrow w_i \in \vseg (\at)) \wedge \model, w_{m+1} \vDash_{a } \at_x\}|}{|\{a  : a  \in \Pg \; \wedge \; \forall \at, i(i\leq m \to (w_i \in a (\at) \leftrightarrow w_i \in \vseg (\at)))\}|}
= \frac{| W ^{\vseg}(\fOne) | - | {\Rexp} _{w_m}^{\vseg}(\fOne) |}{| W | - | {\Rexp} _{w_m} | }.
$
$\{a  : a  \in \Pg \; \wedge \; \forall \at, i(i\leq m \to (w_i \in a (\at) \leftrightarrow w_i \in \vseg (\at)))\}$ contains all the assignments in $\Pg$ that behave like $\vseg$ in relation to $w_1, \ldots , w_m$.
We can count them by checking how many times each atom $p_i$ is satisfied in $W$ according to $\Pg$ and how many times it has already been satisfied in the worlds accessible by $w_m$.
Indeed, each of their elements  corresponds to a distribution of the worlds after $w_m$ to the remaining occurrences of the atoms. 
More concretely, the assignments  correspond to a choice of $| W ^{\vseg}(\at _1)| - | {\Rexp} _{w_m}^{\vseg}(\at_1)|$ worlds in which $\at_1$ is verified taken from the worlds \textit{after} $w_m$ (that is, $| W| - | {\Rexp} _{w_m} |$), a choice of $| W ^{\vseg}(\at _2)| - | {\Rexp} _{w_m}^{\vseg}(\at _2)|$ worlds in which $p_2$ is verified taken from the worlds \textit{after} $w_m$ that do not validate $p_1$ $| W| - | {\Rexp} _{w_m} | - (| W ^{\vseg}(\at_1)| - | {\Rexp} _{w_m}^{\vseg}(\at_1)|)$, and so on.
This leads to the following number of assignments in the denominator:\footnote{With the convention that for every formula $\fOne$, $\sum ^{0}_{j=1} \fOne = 0$.}
\end{sloppypar}
\begin{equation}
\label{eq:denum}
\prod _{i=1}^{s} {{| W | - | {\Rexp} _{w_m} | - \sum ^{i-1}_{j=1}(\fr (p_j)\times | W | - | {\Rexp} _{w_m}^{\vseg}(p_j) |) }\choose{\fr (p_i)\times | W | - | {\Rexp} _{w_m}^{\vseg}(p_i) |}}
\end{equation}
where:
\begin{itemize}
\itemsep0em
\item $| W | - | {\Rexp} _{w_m} |$ are the worlds after $w_m$;
\item for every atom $\at_i$, $\sum ^{i}_{j=1}(\fr (p_j)\times | W | - | {\Rexp} _{w_m}^{\vseg}(p_j)|)$ is the number of worlds after $w_m$ that verify atoms $\at_1 , \ldots , \at_{i-1}$;
\item $\fr (\at_i)\times | W |$ is the total number of worlds that verify $\at_i$ and so $\fr (\at_i)\times | W | - | {\Rexp} _{w_m}^{\vseg}(\at_i) |$ is the number of the worlds that verify $\at_i$ after $w_m$.
\end{itemize}
Following the same argument, we can obtain the number of assignments in the numerator.
These corresponds to all possible choices of assignments of atoms compatible with $\Pg$, when $w_1, \ldots , w_m$ get assigned to atoms as in $\vseg$ and $w_{m+1}$ satisfies $\at_x$.
Hence, the numerator contains the following number of assignments:
\begin{equation}
\label{eq:num}
\prod _{i=1}^{s} {{| W | - | {\Rexp} _{w_{m+1}} | - \sum ^{i-1}_{j=1}(\fr (p_j)\times | W | - | {\Rexp} _{w_{m+1}}^{\vseg '}(p_j) |) }\choose{\fr (p_i)\times | W | - | {\Rexp} _{w_{m+1}}^{\vseg '}(p_i) |}}
\end{equation}
\noindent
where $\vseg '$ behaves like $\vseg$ on $w_1, \ldots , w_m$ and assigns $w_{m+1}$ to $p_x$.
Developing the binomial coefficient, and dividing equation~\ref{eq:num} for equation~\ref{eq:denum}, we obtain the desired:
$
\frac{\fr (p_x) \times | W| - | {\Rexp} _{w_m}^{\vseg}(p_x) |}{| W | - | {\Rexp} _{w_m} | }.
$
While proof is long and tedious, the overall proof strategy becomes obvious if we check the case for only two atoms $p_1$ and $p_2$ (with $p_x = p_1$).
In this case, equation~\ref{eq:denum} becomes:
\begin{equation}
\label{eq:denumSimple}
{{| W | - | {\Rexp} _{w_m} | }\choose{\fr (p_1)\times | W | - | {\Rexp} _{w_m}^{\vseg}(p_1) |}}
\end{equation}
\noindent
Indeed, the second factor 
$
{{| W | - | {\Rexp} _{w_m} | - ((\fr (p_1)\times | W | - | {\Rexp} _{w_{m+1}}^{\vseg }(p_1)) }\choose{\fr (p_2)\times | W | - | {\Rexp} _{w_m}^{\vseg}(p_2) |}}
$
gives only one combination possible, since all the worlds after $w_m$ that do not verify $p_1$ (i.e. $| W | - | {\Rexp} _{w_m} | - (\fr (p_1)\times | W | - | {\Rexp} _{w_{m+1}}^{\vseg}(p_1)|)$) verify $p_2$ (and so correspond to $\fr (p_2)\times | W | - | {\Rexp} _{w_m}^{\vseg}(p_2) |$).
Eq.~\ref{eq:num} becomes:
\begin{equation}
\label{eq:numSimple}
{{| W | - | {\Rexp} _{w_{m+1}} | }\choose{\fr(p_1)\times | W | - | {\Rexp} _{w_{m+1}}^{\vseg '}(p_1) |}}
\end{equation}
\begin{sloppypar}
\noindent
where it holds that $| W | - | {\Rexp} _{w_{m+1}} | = | W | - | {\Rexp} _{w_{m}} | - 1 $, and $\fr (p_1)\times | W | - | {\Rexp} _{w_{m+1}}^{\vseg '}(p_1) | = \fr (p_1)\times | W | - | {\Rexp} _{w_m}^{\vseg}(p_1) | - 1$.
Let us use $\overrightarrow{{\Rexp} _{w_m}^{\vseg}(p_1)}$ as an abbreviation for $\fr (p_1)\times | W | - | {\Rexp} _{w_m}^{\vseg}(p_1) |$.
Hence, if we develop the binomial coefficient and divide Eq.~\ref{eq:numSimple} for Eq.~\ref{eq:denumSimple}, we obtain:
\tiny
$
\frac{(| W | - | {\Rexp} _{w_{m}} |-1)!}{(|\overrightarrow{{\Rexp} _{w_m}^{\vseg}(p_1)}| - 1)!(| W | - | {\Rexp} _{w_{m}} |-1 -(|\overrightarrow{{\Rexp} _{w_m}^{\vseg}(p_1)}| - 1))!}
\times
\frac{(|\overrightarrow{{\Rexp} _{w_m}^{\vseg}(p_1)}|)!(| W | - | {\Rexp} _{w_{m}} | -(|\overrightarrow{{\Rexp} _{w_m}^{\vseg}(p_1)}|))!}{(| W | - | {\Rexp} _{w_{m}} |)!}=$

$\frac{(| W | - | {\Rexp} _{w_{m}} |-1)!}{(|\overrightarrow{{\Rexp} _{w_m}^{\vseg}(p_1)}| - 1)!(| W | - | {\Rexp} _{w_{m}} | -|\overrightarrow{{\Rexp} _{w_m}^{\vseg}(p_1)}| )!}
\times
\frac{(|\overrightarrow{{\Rexp} _{w_m}^{\vseg}(p_1)}|) (|\overrightarrow{{\Rexp} _{w_m}^{\vseg}(p_1)}| -1)!(| W | - | {\Rexp} _{w_{m}} | -(|\overrightarrow{{\Rexp} _{w_m}^{\vseg}(p_1)}|))!}{(| W | - | {\Rexp} _{w_{m}} |) (| W | - | {\Rexp} _{w_{m}} | - 1)!}
$
\normalsize
which simplifies to
\tiny
$
\frac{\fr (p_1)\times | W | - | {\Rexp} _{w_m}^{\vseg}(p_1) |}{| W | - | {\Rexp} _{w_{m}} |}.
$
\end{sloppypar}
\end{proof}
\normalsize

\end{document}